\documentclass[10pt,journal]{IEEEtran}

\usepackage{amssymb}
\usepackage{amsmath}
\usepackage{cite}
\usepackage{url}
\usepackage{xcolor}
\usepackage{cite,graphicx,amsmath,amssymb}
\usepackage{subfigure}
\usepackage{fancyhdr}
\usepackage{mdwmath}
\usepackage{mdwtab}
\usepackage{caption}
\usepackage{amsthm}
\usepackage{setspace}
\usepackage{hyperref}
\usepackage{algorithm}
\usepackage{algorithmic}
\usepackage{multirow}
\usepackage{makecell}
\usepackage{mathtools}
\hypersetup{colorlinks=false}

\graphicspath{{./src/img/}}

\newtheorem{theorem}{Theorem}

\newtheorem{lemma}{Lemma}

\newtheorem{corollary}{Corollary}

\allowdisplaybreaks
\title{Simultaneously Transmitting and Reflecting Surface (STARS) for Terahertz Communications}

\author{
        Zhaolin~Wang,~\IEEEmembership{Graduate Student Member,~IEEE,}
        Xidong~Mu,~\IEEEmembership{Member,~IEEE,} \\
        Jiaqi~Xu,~\IEEEmembership{Graduate Student Member,~IEEE,}
        and Yuanwei~Liu,~\IEEEmembership{Senior Member,~IEEE}
        
\thanks{The authors are with the School of Electronic Engineering
and Computer Science, Queen Mary University of London, London E1 4NS, U.K. (e-mail: zhaolin.wang@qmul.ac.uk, xidong.mu@qmul.ac.uk, jiaqi.xu@qmul.ac.uk, yuanwei.liu@qmul.ac.uk).}
\vspace{-0.8cm}
}

\begin{document}

\maketitle
\begin{abstract}
    A simultaneously transmitting and reflecting surface (STARS) aided terahertz (THz) communication system is proposed. A novel power consumption model is proposed that depends on the type and resolution of the STARS elements. The spectral efficiency (SE) and energy efficiency (EE) are maximized in both narrowband and wideband THz systems by jointly optimizing the hybrid beamforming at the base station (BS) and the passive beamforming at the STARS. 1) For narrowband systems, independent phase-shift STARSs are investigated first. The resulting complex joint optimization problem is decoupled into a series of subproblems using penalty dual decomposition. Low-complexity element-wise algorithms are proposed to optimize the analog beamforming at the BS and the passive beamforming at the STARS. The proposed algorithm is then extended to the case of coupled phase-shift STARS. 2) For wideband systems, the spatial wideband effect at the BS and STARS leads to significant performance degradation due to the beam split issue. To address this, true time delayers (TTDs) are introduced into the conventional hybrid beamforming structure for facilitating wideband beamforming. 
    An iterative algorithm based on the quasi-Newton method is proposed to design the coefficients of the TTDs. Finally, our numerical results confirm the superiority of the STARS over the conventional reconfigurable intelligent surface (RIS). It is also revealed that i) there is only a slight performance loss in terms of SE and EE caused by coupled phase shifts of the STARS in both narrowband and wideband systems, and ii) the conventional hybrid beamforming achieves comparable SE performance and much higher EE performance compared with the full-digital beamforming in narrowband systems but not in wideband systems, where the TTD-based hybrid beamforming is required for mitigating wideband beam split.
\end{abstract}

\section{Introduction}
The sixth generation (6G) wireless communication systems are anticipated to support a minimum peak data rate of one terabit per second (Tbps) to enable the development of novel applications, including virtual reality, vehicle-to-everything, Internet of Things, and Metaverse \cite{zhang20196g}. In this context, communication over the terahertz (THz) band, which is situated in the frequency range of 0.1-10 THz, is considered a promising technique for 6G as it provides a broad communication bandwidth in the order of tens of gigahertz (GHz) \cite{akyildiz2014terahertz, shafie2022terahertz, akyildiz2022terahertz}. However, THz signals suffer significant propagation pass loss due to their very high frequencies, which limits communication distance. Therefore, the massive multiple-input multiple-output (MIMO) technique is a crucial enabler for THz communication \cite{akyildiz2022terahertz}. Specifically, the ultra-massive antenna array can be implemented thanks to the extremely small wavelength of THz signals, allowing for the generation of fine beams to compensate for significant pass loss in the THz band.

Simultaneous transmitting and reflecting surface (STARS) is another promising technique for 6G. Unlike the conventional reconfigurable intelligent surface (RIS) that can only reflect the incident signal and thus lead to the \emph{half-space} coverage \cite{huang2019reconfigurable}, the STARS can simultaneously transmit and reflect the incident signal into both sides of the surface, resulting in the \emph{full-space} coverage \cite{liu2021star, mu2021simultaneously}. Therefore, STARS provides more degrees of freedom for manipulating the signal propagation and thereby enhances the design flexibility of the wireless network.

\subsection{Prior Works}
As previously discussed, the massive MIMO technique is crucial for THz communications, enabling accurate beamforming to combat significant pathloss. However, in contrast to sub-6 GHz communications, the full-digital beamforming architecture, where each antenna has a dedicated radio frequency (RF) chain, becomes infeasible in high-frequency bands such as millimeter-wave (mmWave) and THz, due to its high cost and power consumption \cite{han2021hybrid}. Consequently, researchers have extensively studied a hybrid analog and digital beamforming architecture, which comprises only a few RF chains for digital beamforming and numerous low-cost phase shifters (PSs) for analog beamforming \cite{yu2016alternating, sohrabi2016hybrid, ni2017near, zhu2017hybrid, shi2018spectral}. In particular, the authors of \cite{yu2016alternating} proposed a series of alternating minimization algorithms for hybrid beamforming design, aiming to minimize the matching error between hybrid beamforming and optimal full-digital beamforming. The authors of \cite{sohrabi2016hybrid} analyzed the minimum number of RF chains required for hybrid beamforming to achieve performance comparable to that of full-digital beamforming. Furthermore, the matrix decomposition was employed for hybrid beamforming design in \cite{ni2017near} and \cite{zhu2017hybrid}, leading to relatively low complexity. As a further advance, a penalty-based algorithm with provable optimality was developed in \cite{shi2018spectral} to investigate the performance limit of hybrid beamforming.

Nevertheless, the above hybrid beamforming designs have mainly focused on narrowband systems and may not meet the requirements of THz communications with large spectrum resources. As a result, research efforts have been directed towards investigating the hybrid beamforming design in THz wideband systems \cite{yuan2020hybrid, yuan2022cluster, gao2021wideband, dai2022delay}. For instance, a two-stage wideband hybrid beamforming design was proposed in \cite{yuan2020hybrid} for multi-carrier systems over frequency selective fading channels. To overcome the connectivity limitation posed by the number of RF chains, the authors of \cite{yuan2022cluster} introduced a cluster-based hybrid beamforming design to serve multiple users in the same beam with different subcarriers. However, neither of the approaches considered the significant beam split effect that occurs in wideband THz communications. To address this challenge, two hybrid beamforming designs based on virtual sub-array and true time delayers (TTDs) were presented in \cite{gao2021wideband}. Furthermore, the authors of \cite{dai2022delay} proposed a delay-phase beamforming design, which also employs TTDs to eliminate the beam split effect.

In addition, the high susceptibility of THz communications to obscuration may lead to transmission unreliability caused by blockages. One possible solution to address this issue is through the use of RISs, which can establish an additional line-of-sight path, thus enhancing the performance of THz communications \cite{chen2021intelligent}. For instance, the authors of \cite{ning2021terahertz} proposed a cooperative beam training scheme for RIS-aided THz communications and designed the hybrid beamforming based on the training results. To address the effects of imperfect channel state information, a robust hybrid design was developed for RIS-aided THz communications in \cite{hao2021robust}. Furthermore, the authors of \cite{wan2021terahertz} propose to exploit holographic RISs in THz communications to enhance the array gain. More recently, the beam split effect in RIS-aided THz communications was investigated in \cite{yan2022beamforming} and \cite{su2023wideband}.

\begin{figure}[t!]
    \centering
    \includegraphics[width=0.4\textwidth]{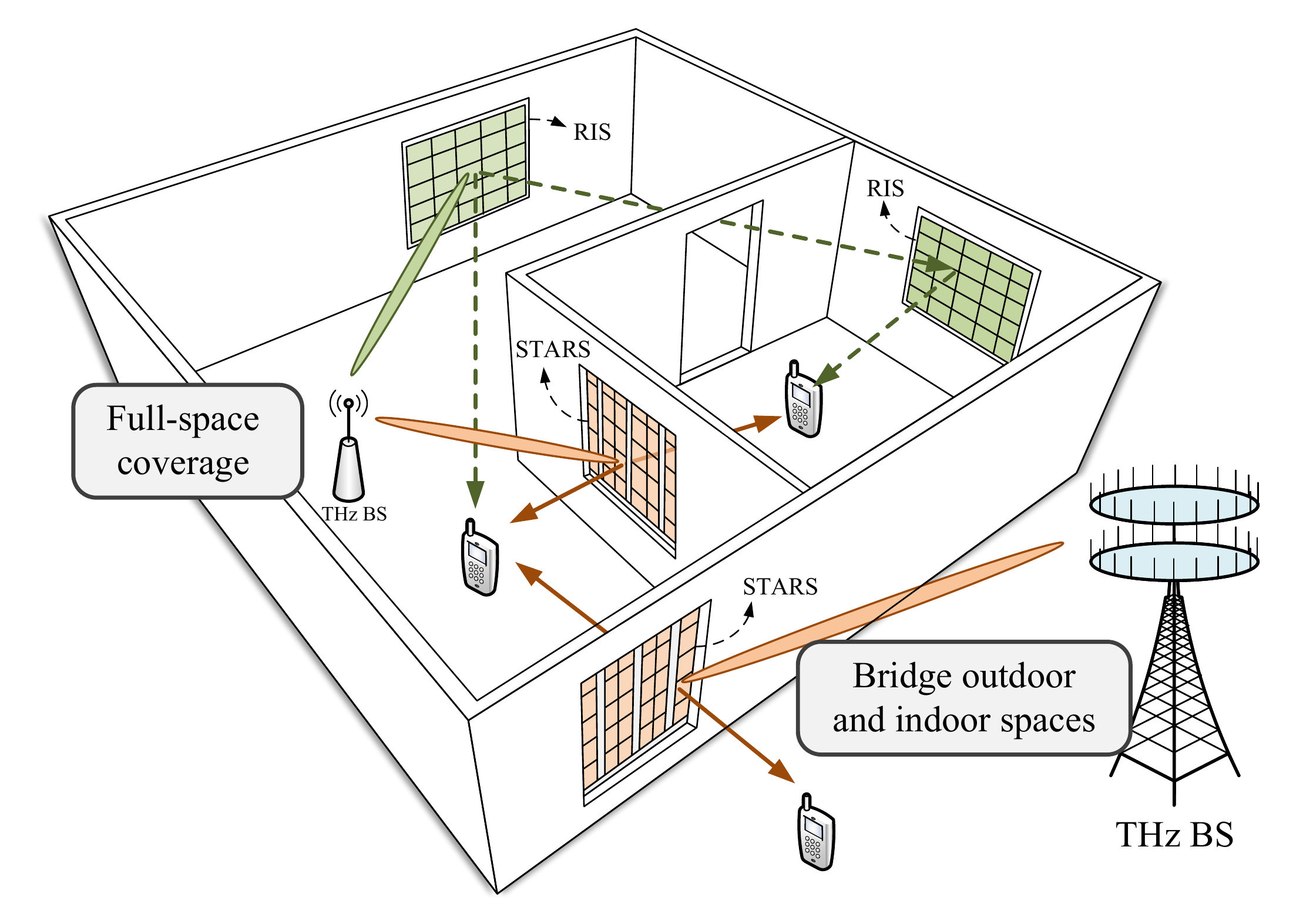}
    \caption{Applications of STARS in THz communications.}
    \label{fig:system_model}
\end{figure}

\subsection{Motivations and Contributions}
Although RISs have the potential to solve the blockage problem in THz communications, their hardware limitations restrict their flexibility and effectiveness. Specifically, conventional RISs can only reflect the incident signals, thus requiring the transmitter and receiver to be located on the same side of RISs. As a result, multiple RISs may be needed to overcome blockages and realize full-space coverage. For instance, in indoor THz communications, as depicted in Fig. \ref{fig:system_model}, two conventional RISs are needed to cover the entire indoor space. STARSs are promising to overcome the aforementioned limitations and provide more benefits. On the one hand, STARSs are more efficient in achieving full-space coverage. As shown in Fig. \ref{fig:system_model}, a single STARS can aid the THz base station (BS) in covering the entire indoor space. On the other hand, STARSs can bridge disconnected spaces, such as indoor and outdoor spaces, which is impossible for conventional RISs. Therefore, there is a natural link between STARSs and THz communications, which motivates further exploration of STARSs in THz communications.

Based on different spectrum allocation schemes, such as multi-band-based and multi-transmission-window-based schemes, the system bandwidth of THz communications can be either reasonably small or extremely large \cite{shafie2022terahertz}. Furthermore, based on different hardware implementations, the transmission and reflection phase shifts of STARSs can be either independent or coupled \cite{xu2022star_couple}. Therefore, to explore the full potential of STARSs in THz communications, we investigate the performance of both independent and coupled phase-shift STARSs in both narrowband and wideband THz communication systems. The main contributions of this paper can be summarized as follows:

\begin{itemize}
    \item We propose a novel STARS-aided THz communication system and evaluated its performance in terms of spectral efficiency (SE) and energy efficiency (EE). We for the first time develop a power consumption model for STARS, which takes into account the types and resolutions of the STARS elements. Furthermore, we analyze the THz channels in both narrowband and wideband systems. Based on the different characteristics of narrowband and wideband THz channels, we investigate the joint design of hybrid beamforming at the BS and the passive beamforming at the STARS for both systems.
    
    \item For narrowband systems, we formulate a general optimization problem that aims to maximize SE and EE. Specifically, for the independent phase-shift STARS, we propose a double-loop iterative algorithm using penalty dual decomposition (PDD) \cite{shi2020penalty} to solve the optimization problem. In particular, we develop element-wise algorithms that admit optimizing the analog beamforming at the BS and the passive beamforming at the STARS with low computational complexity. Finally, we extend the proposed algorithm to cover scenarios with the coupled phase-shift STARS.
    
    \item For wideband systems, we focus on alleviating the impact of beam split caused by the spatial wideband effect at the BS and STARS. We introduce TTDs into the conventional hybrid beamforming structure for facilitating wideband beamforming. Then, we propose an iterative algorithm based on the quasi-Newton method to optimize the coefficients of TTDs.

    \item Our numerical results unveil that 1) the coupled phase-shift STARS only leads to a slight performance degradation compared to the independent one, in both narrowband and wideband systems; 2) In narrowband systems, the performance of hybrid beamforming is comparable to the full-digital beamforming; and 3) in wideband systems, TTD-based hybrid beamforming achieves similar performance to the full-digital beamforming, while conventional hybrid beamforming causes significant performance degradation especially when the size of STARS is large.
\end{itemize}

\subsection{Organization and Notations}
The remainder of this paper is organized as follows. Section \ref{sec:system_model} presents the system model of the proposed STARS-aided THz communication system. Then, Sections \ref{sec:narrow} and \ref{sec:wide} investigate the joint designs of hybrid beamforming at the BS and passive beamforming at the for narrowband and wideband systems, respectively. Section \ref{sec:result} provides numerical results to verify the effectiveness of the proposed designs. Finally, Section \ref{sec:conclusion} concludes the paper.

\emph{Notations:}
Scalars, vectors, and matrices are denoted by the lower-case, bold-face lower-case, and bold-face upper-case letters, respectively; 
$\mathbb{C}^{N \times M}$ denotes the space of $N \times M$ complex matrices;
$a^*$ and $|a|$ denote the conjugate and the magnitude of scalar $a$, respectively;  
$(\cdot)^T$, $(\cdot)^H$, $\|\cdot\|$, $\|\cdot\|_F$, and $\mathrm{tr}(\cdot)$ denote the transpose, conjugate transpose, norm, Frobenius norm, and trace respectively; 
$\mathrm{blkdiag}(\mathbf{A})$ denotes a block diagonal matrix in which the diagonal blocks are the columns of $\mathbf{A}$;
$\mathbf{I}_N$ denotes the $N \times N$ identity matrix;   
$[\mathbf{A}]_{i,j}$ denotes the entry of the matrix $\mathbf{A}$ at the $i$-th row and $j$-th column;  
$[\mathbf{a}]_{i:j}$ denotes the vector composed of the $i$-th to the $j$-th entries of the vector $\mathbf{a}$;  
$\mathbb{E}[\cdot]$ denotes the statistical expectation; 
$\mathrm{Re}\{\cdot\}$ denotes the real component of a complex number;
$\mathcal{CN}(\mu, \sigma^2)$ denotes the distribution of a circularly symmetric complex Gaussian random variable with mean $\mu$ and variance $\sigma^2$;
$\mathcal{U}(a,b)$ denotes the uniform distribution between $a$ and $b$;
$\lceil \cdot \rceil$ denotes the ceiling function.

\section{System Model} \label{sec:system_model}

We consider a STARS-aided THz communication system, which consists of an $N$-antenna BS with a uniform linear array (ULC), an $M$-element STARS with a uniform planar array (UPA), and $K$ single-antenna users whose indices are collected in $\mathcal{K}$.  Without loss of generality, we assume that the users in subset $\mathcal{K}_t = \{1,\dots,K_0\}$  are located on the transmission side, and the users in subset $\mathcal{K}_r = \{K_0+1,\dots,K\}$ are located on the reflection side. Due to the high susceptibility of THz communications to obscuration, the direct links between the BS and users are assumed to be blocked. Furthermore, it is assumed that the channels and path angles have been acquired through the angle-based channel estimation method \cite{liu2021cascaded}.

\subsection{Signal Model for STARS}
In this work, we consider patch-array-based STARSs \cite{xu_vtmag}. The STARS elements excited by the incident signal are capable of radiating signals into both transmission and reflection spaces, which are referred to as transmitted signals and reflected signals, respectively. Let $s_m \in \mathbb{C}$ denote the incident signal at the $m$-th element. Then, the corresponding transmitted signal $t_n \in \mathbb{C}$ and reflected signal $r_n \in \mathbb{C}$ is given by \cite{xu2021star}
\begin{equation}
  t_m = \beta_{t,m} e^{j \phi_{t,m}} s_m, \quad r_m = \beta_{r,m} e^{j \phi_{r,m}} s_m,
\end{equation}
where $\beta_{t,m}, \beta_{r,m} \in [0,1]$ are the amplitude coefficients for transmission and reflection and $\phi_{t,m}, \phi_{r,m} \in [0, 2\pi]$ are the corresponding phase shifts introduced by the $m$-th elements. In this paper, we consider both independent and coupled phase-shift models for the STARS. In particular, for independent phase-shift STARSs, the law of energy conservation needs to be satisfied, which is given by 
\begin{equation} \label{eqn:ideal_STAR}
  \beta_{t,m}^2 + \beta_{r,m}^2 = 1, \forall m \in \mathcal{M},
\end{equation}
where $\mathcal{M} = \{1,\dots,M\}$. Moreover, the phase shifts for transmission and reflection can be independently adjusted. However, the active or lossy elements are required to achieve independent control of phase shifts, which can significantly increase manufacturing costs. For the low-cost passive lossless STARS, the electric and magnetic impedances of each element should be purely imaginary. Under such conditions, the transmission and reflection phase shifts of STARSs are coupled, leading to the following constraints \cite{xu2022star_couple}:
\begin{align} \label{eqn:coupled_STAR} 
  \cos(\phi_{t,m} - \phi_{r,m}) = 0, \forall m \in \mathcal{M}.
\end{align}
The above constraint implies that if $\phi_{t,m}$ is fixed, $\phi_{r,m}$ can only be selected from a finite set such that $\phi_{t,m} - \phi_{r,m}$ is $\pi/2$ or $3\pi/2$, and vise versa.

\begin{figure}[t!]
  \centering
  \includegraphics[width=0.4\textwidth]{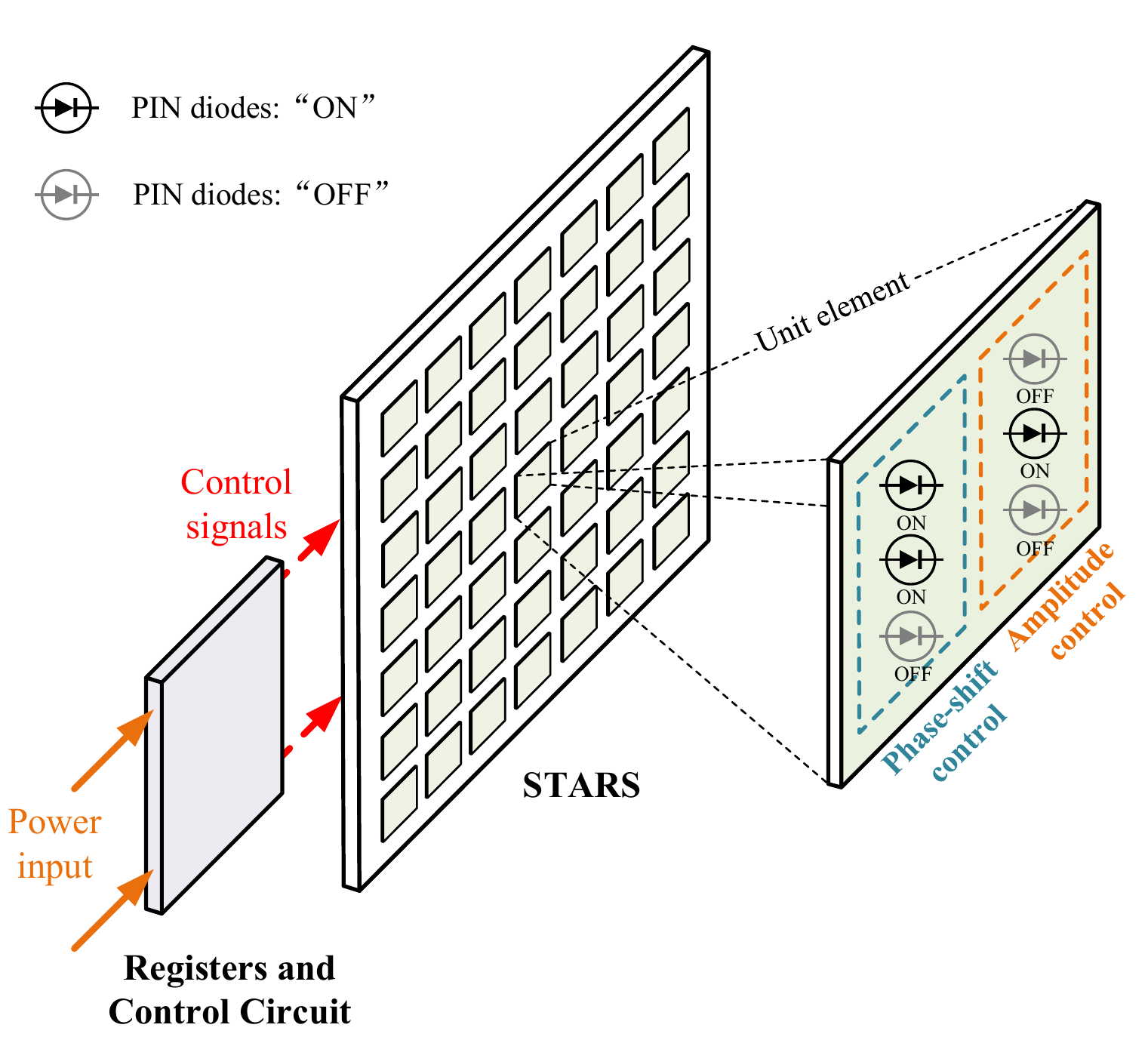}
  \caption{Hardware structure of the STARS}
  \label{fig:pwoer_model}
\end{figure}

\subsection{Proposed Power Consumption Model for STARSs} \label{sec:STAR_power}
For patch-array-based STARSs, each element can accommodate positive-intrinsic-negative (PIN) diodes to configure the different states, as illustrated in Fig. \ref{fig:pwoer_model}. The power consumption of STARSs is composed of two parts, namely the static power consumption of the control circuit (e.g., field-programmable gate array (FPGA) board connected to the PIN diodes) and the dynamic power consumption of each element \cite{9206044}. In principle, the static power consumption is independent of the operating states of the STARS elements. On the contrary, the dynamic power consumption depends on the different states of the element as well as the number of quantization levels of each element. In the following, we present the power consumption models for both independent phase-shift and coupled phase-shift STARSs.

\subsubsection{Independent Phase-shift STARSs}
Independent phase-shift STARSs are only subject to the energy conservation constraint as given in \eqref{eqn:ideal_STAR}. Thus, as shown in Fig. \ref{fig:pwoer_model}, we need three sets of PIN diodes to control the power splitting ratio ($\beta_{t,m}/\beta_{r,m}$), the transmission phase-shift ($\phi_{t,m}$), and the reflection phase-shift ($\phi_{r,m}$), respectively. 
If the number of quantization levels assigned for the power splitting ratio and phase shifts are $L_\beta$ and $L_\phi$, respectively, then, the total number of PIN diodes required is $\lceil\log_2L_\beta + 2\log_2L_\phi\rceil$. 
As a result, we formulate the power consumption of the independent phase-shift STARS as follows:
\begin{equation} \label{eqn:STAR_power}
    P_{\mathrm{STAR}}^i = \frac{1}{2}\lceil\log_2L_\beta + 2\log_2L_\phi\rceil M P_{\mathrm{PIN}} + P_{\mathrm{circ}},
\end{equation}
where $\frac{1}{2}\lceil\log_2L_\beta + 2\log_2L_\phi\rceil$ is the average number of PIN diodes in the ``ON'' state per element, $P_{\mathrm{PIN}}$ denotes the power consumption per PIN diode, and $P_{\mathrm{circ}}$ denotes the static power consumption of the control circuit.

\subsubsection{Coupled Phase-shift STARSs}
Coupled phase-shift STARSs are subject to the energy conservation constraint and the coupled phase-shift constraint in \eqref{eqn:coupled_STAR}. Thus, for each element, we need two sets of PIN diodes to control the power splitting and phase-shift for $\phi_{t,m}$, as well as an auxiliary PIN diode to determine whether $\phi_{t,m}-\phi_{r,m}$ is $\pi/2$ or $3\pi/2$ \cite{xu2022star_couple}. This is to say that for coupled phase-shift STARSs, the total number of PIN diodes required is $\lceil\log_2L_\beta + \log_2L_\phi+1\rceil$.
As a result, we formulate the power consumption of coupled phase-shift STARSs STARSs as follows:
\begin{equation} \label{eqn:STAR_power_coupled}
    P_{\mathrm{STAR}}^c = \frac{1}{2}\lceil\log_2L_\beta + \log_2L_\phi+1\rceil M P_{\text{PIN}} + P_{\text{circ}},
\end{equation}
Theoretically, to achieve the ideal case of continuous phase shifts, an infinite number of states are required. However, in practical cases, continuous phase shifts are obtained with tolerable phase-shift errors. For example, we consider the maximum tolerable error to be $1^\circ$, then the corresponding ``continuous phase-shift'' can be achieved with $\lceil\log_2{180} \rceil = 8$ PIN diodes per element. A similar result holds for the control bits of power splitting amplitudes.
\subsection{THz Channel Model} \label{sec:channel_model}
We adopt a ray-tracing-based channel model to incorporate the limited scattering characteristics of THz channels in both narrowband and wideband systems. Consider a THz channel with $L$ paths between the BS and STARS and $L_k$ paths between the STARS and user $k$. For the channel related to the $i$-th path between the BS and STARS and the $j$-th path between the STARS and user $k$, let $\alpha_{i,j,k}$ denote as the complex path gain, $\tau_{i,j,k}$ denote the path delay, $\psi_{j,k}^t$, $\vartheta_i^t$,  $\psi_i^r$, and $\vartheta_{j,k}^r$ denote the azimuth angle of departure, elevation angle of departure, azimuth angle of arrival, and elevation angle of arrival at the STARS, respectively, and $\varphi_i$ denote the angle of departure at the BS. Let $p(t)$ denote the pulse shaping function, $W$ denote the system bandwidth, and $f$ denote the carrier frequency. Then, the corresponding delay-$q$ baseband channel $\mathbf{h}_{i,j,k}[q] \in \mathbb{C}^{1 \times N}, \forall k \in \mathcal{K}_{\chi},$ is given by \cite{tse2005fundamentals, priebe2013stochastic, rappaport2019wireless}
\begin{align} \label{eqn:channel_path}
  \mathbf{h}_{i,j,k}&[q] = \alpha_{i,j,k} p(q - W \tau_{i,j,k}) \nonumber \\
  &\times \mathbf{a}^H (f, \psi_{j,k}^t, \vartheta_{j,k}^t) \mathbf{\Theta}_{\chi} \mathbf{a}(f, \psi_i^r, \vartheta_i^r) \mathbf{b}^H(f, \varphi_i),
\end{align}  
where $\mathbf{\Theta}_{\chi} = \mathrm{diag}([\beta_{\chi,1} e^{j \varphi_{\chi,1}},\dots,\beta_{\chi,M} e^{j \varphi_{\chi,M}}]^T), \forall \chi \in \{t,r\},$ denotes the transmission/reflection coefficient matrix of the STARS, $\mathbf{a}(f, \psi, \vartheta)$ denotes the array response vector of the UPA at the STARS, and $\mathbf{b}(f, \varphi)$ denotes the array response vector of the ULA at the BS. Specifically, assuming the dimension of the UPA at the STARS to be $M_h \times M_v$, the array response vector can be modeled as 
\begin{align}
  &\mathbf{a} (f, \psi, \vartheta) \nonumber \\
  & = [1, e^{-j \frac{2\pi f}{c} d \sin \psi \sin \vartheta  } , \dots, e^{-j \frac{2 \pi f}{c} (M_h-1) d \sin \psi \sin \vartheta }]^T \nonumber \\
  & \qquad \quad \otimes [1, e^{-j \frac{2 \pi f}{c} d \cos \vartheta }, \dots, e^{-j \frac{2 \pi f}{c} (M_v-1) d \cos \vartheta }]^T,
\end{align}
where $c$ and $d$ denote the speed of light and antenna spacing, respectively. Similarly, the array response vector of the ULA at the BS is given by    
\begin{equation} \label{eqn:array_response_BS}
  \mathbf{b} (f, \varphi) = [1, e^{-j \frac{2 \pi f}{c} d \sin \varphi }, \dots, e^{-j \frac{2 \pi f}{c} (N-1) d \sin \varphi }]^T.
\end{equation}
Furthermore, the overall complex path gain can be rewritten as $\alpha_{i,j,k} = \bar{\alpha}_i \tilde{\alpha}_{j,k} \sqrt{G_r G_t}$. More specifically, $\bar{\alpha}_i$ and $\tilde{\alpha}_{j,k}$ denote the complex path gains between the BS and the STARS and between the STARS and user $k$, respectively. $G_t$ and $G_r$ denote the transmit and receive antenna gain, respectively. The amplitude of $\bar{\alpha}_i$ and $\tilde{\alpha}_{j,k}$ relies on the path loss. According to \cite{jornet2011channel}, the path loss $L(f, D)$ in the THz band involves spreading loss and absorption loss, which is given by 
\begin{align} \label{eqn:pathloss}
  L(f, D) [\mathrm{dB}] = &L_{\mathrm{spread}}(f, D) [\mathrm{dB}] + L_{\mathrm{absorption}}(f, D) [\mathrm{dB}] \nonumber \\
  = & 20 \log_{10} \left( \frac{4 \pi f D}{c} \right) + k(f) D 10\log_{10} e.
\end{align}
Here, $D$ denotes the path length and $k(f)$ is the frequency-dependent medium absorption coefficient.

Given the delay-$q$ channel $\mathbf{h}_{i,j,k}[q]$ of each path, the overall channel for user $k$ at time $n$ is given by    
\begin{equation}
  \mathbf{h}_k[n] = \sum_{q=0}^{Q_k-1} \sum_{i =1}^L \sum_{j = 1}^{L_k} \mathbf{h}_{i,j,k}[q] \delta[n-q],
\end{equation}
where $Q_k$ denotes the maximum resolvable delays at user $k$ and $\delta[\cdot]$ denotes the unit impulse function. The value of $Q_k$ is determined by the relationship between the system bandwidth and the coherence bandwidth $W_{corr,k} = 1/[(\max_{i,j} \tau_{i,j,k}) - (\min_{i,j} \tau_{i,j,k})]$. Therefore, in the following, we further analyze the overall channel $\mathbf{h}_k[n]$ in narrowband and wideband systems, respectively

\subsubsection{Narrowband System}
In narrowband systems, the system bandwidth is assumed to be much smaller than the coherence bandwidth, i.e., $W \ll W_{corr,k}$ \cite{tse2005fundamentals}. For the purpose of exposition, we set the time of the first arrival path as the reference time, i.e., $\min_{i,j} \tau_{i,j,k} = 0$. Hence, the coherence bandwidth can be simplified as $W_{corr,k} = 1/\max_{i,j} \tau_{i,j,k}$. Based on the narrowband condition $W \ll W_{corr,k}$, it holds that $W \tau_{i,j,k} \approx 0, \forall m,n$, resulting in $p(q - W \tau_{i,j,k}) \approx p(q)$. Generally, the pulse-shaping function can be the optimal sinc function or the practical raise-cosine function, which has the property of rapid decay, i.e., $p(q) = 1$ when $q = 0$ and $p(q) \approx 0$ when $q \ge 1$. Therefore, the value of delay-$q$ channel $\mathbf{h}_{m,n,q}$ is almost zeros when $q \ge 1$. Therefore, the overall channel for user $k, \forall k \in \mathcal{K}_{\chi},$ can be simplified into the following form:    
\begin{align}
  \mathbf{h}^{n}_{k}[n] = &\sum_{i=1}^L \sum_{j=1}^{L_k} \mathbf{h}_{i,j,k} [0] \delta[n] =  \mathbf{v}_k^n \mathbf{\Theta}_{\chi} \mathbf{G}^n \delta[n],
\end{align} 
where 
\begin{align}
  \mathbf{G}^n = &\sum_{i=1}^L \sqrt{G_t} \bar{\alpha}_i \mathbf{a}(f, \psi_i^r, \vartheta_i^r) \mathbf{b}^H(f, \varphi_i), \\
  \mathbf{v}_k^n = &\sum_{j=1}^{L_k} \sqrt{G_r} \tilde{\alpha}_{j,k} \mathbf{a}^H (f, \psi_{j,k}^t, \vartheta_{j,k}^t).
\end{align}

\subsubsection{Wideband System}
In wideband systems, the system bandwidth is assumed to be comparable to or much larger than the coherence bandwidth. In this case, the condition $W \tau_{i,j,k} \approx 0$ no longer holds. Therefore, unlike narrowband systems where a single delay channel is sufficient to represent the overall channel, multiple delay channels are required for wideband systems, i.e., $Q_k\ge 1$. Such a channel is referred to as a frequency-selective channel, which causes inter-symbol interference since each user $k$ receives overlapped transmit signals with different delays at each time index $n$. The multi-carrier orthogonal frequency-division multiplexing (OFDM) technique is typically exploited to address this issue, where the signal is transformed into the frequency domain using discrete Fourier transform (DFT). Let $M_c$ denote the number of subcarriers in the OFDM system and $f_c$ denote the central frequency. The frequency for subcarrier $m$ is thus given by $f_m = f_c + \frac{W(2m-1-M_c)}{2M_c}$. Then, the overall channel for user $k, \forall k \in \mathcal{K}_{\chi},$ at subcarrier $m$ obtained by DFT is given by 
\begin{align} \label{eqn:overall_channel}
  \mathbf{h}_{m,k}^w = & \sum_{q=0}^{Q_k-1} \sum_{i =1}^L \sum_{j = 1}^{L_k} \mathbf{h}_{i,j,k}[f_m, q] e^{\frac{-j 2 \pi m q}{M_c}} = \mathbf{v}_{m,k}^w \mathbf{\Theta}_{\chi} \mathbf{G}_{m}^w,
\end{align}  
where 
\begin{align}
  \mathbf{G}_m^w = & \sum_{i=1}^L \bar{\alpha}_i^w \mathbf{a}(f_m, \psi_i^r, \vartheta_i^r) \mathbf{b}^H(f_m, \varphi_i), \\
  \mathbf{v}_{m,k}^w = & \sum_{j=1}^{L_k} \tilde{\alpha}_{j,k}^w \mathbf{a}^H (f_m, \psi_{j,k}^t, \vartheta_{j,k}^t), \\
  \bar{\alpha}_i^w \tilde{\alpha}_{j,k}^w = & \sum_{q=0}^{Q_k-1} \alpha_{i,j,k} p(q - W \tau_{i,j,k}) e^{\frac{-j 2 \pi m q}{M_c}}.
\end{align}
Note that in the above expression, we rephrase $\mathbf{h}_{i,j,k} [q]$ for subcarrier $m$ as $\mathbf{h}_{i,j,k}[f_m,q]$, since it is also a function of carrier frequency. To further explain this, we take the array response vector of the ULA at the BS as an example. When manufacturing the antenna arrays, the antenna spacing is usually set as half of the wavelength at the central frequency, i.e., $d = \frac{c}{2 f_c}$. According to \eqref{eqn:array_response_BS}, the array response vector of the ULA at the BS at subcarrier $m$ is given by  
\begin{equation} 
    \mathbf{b} (f_m, \varphi) = [1, e^{-j \pi \frac{f_m}{f_c} \sin \varphi }, \dots, e^{-j \pi (N-1) \frac{f_m}{f_c} \sin \varphi }]^T.
\end{equation}
Therefore, for \emph{frequency-independent} analog beamforming at the BS and the passive beamforming at the STARS, the \emph{frequency-dependent} array response vector can lead to \emph{beam split} effect, which will be detailed in Section \ref{sec:wide}.

\section{Narrowband System} \label{sec:narrow}
In this section, we investigate the narrowband STARS-aided THz communication system with the hybrid beamforming structure with the aim of maximizing its SE and EE. We first propose a PDD-based algorithm for solving the SE and EE maximization problem in the case of independent phase-shift STARSs, which is then extended to the case of coupled phase-shift STARSs.

\subsection{Hybrid Beamforming}
According to \eqref{eqn:pathloss}, the path loss in the THz band can be very large due to very high frequencies. To compensate for the severe path loss, we exploit a massive antenna array with a hybrid beamforming structure, as shown in Fig. \ref{fig:narrow_hybrid}, at the BS to achieve a large array gain. In the hybrid beamforming structure, we assume that there are $N_{\mathrm{RF}}$ RF chains ($N_{\mathrm{RF}} \ll N$). Each RF chain is connected to the $N$ antenna via $N$ phase shifters (PSs). Therefore, there are total $N_{\mathrm{RF}} N$ PSs. Let $\mathbf{F}_{\mathrm{RF}} \in \mathbb{C}^{N \times N_{\mathrm{RF}}}$ denote the analog beamformer achieved by PSs, $\mathbf{F}_{\mathrm{BB}} = [\mathbf{f}_{\mathrm{BB},1}, \dots, \mathbf{f}_{\mathrm{BB},K}] \in \mathbb{C}^{N_{\mathrm{RF}} \times K }$ denote the digital beamformers for $K$ users, and $\mathbf{s}[n] = [s_1[n],\dots,s_K[n]]^T \in \mathbb{C}^{K \times 1}$ denote the information symbols for $K$ users. The transmit signal at the BS is given by 
\begin{equation}
  \mathbf{x}[n] = \mathbf{F}_{\mathrm{RF}} \mathbf{F}_{\mathrm{BB}} \mathbf{s}[n] = \mathbf{F}_{\mathrm{RF}} \sum_{k \in \mathcal{K}} \mathbf{f}_{\mathrm{BB},k} s_k[n].
\end{equation}
Since the PSs can only change phase shifts of signals, each entry of the analog beamformer matrix $\mathbf{F}_{\mathrm{RF}}$ needs to satisfy the following unit-modulus constraint:
\begin{equation}
  \left| [\mathbf{F}_{\mathrm{RF}}]_{i,j} \right| = 1, \forall i, j.
\end{equation}
It is assumed that $\mathbf{s}[n]$ is the independent complex Gaussian signal, i.e., $\mathbb{E}[ \mathbf{s}[n] (\mathbf{s}[n])^H ] = \mathbf{I}_{K}$. Therefore, the covariance matrix of the transmit signal is given by $\mathbf{Q} = \mathbb{E}[ \mathbf{x}[n] (\mathbf{x}[n])^H ] = \mathbf{F}_{\mathrm{RF}} \mathbf{F}_{\mathrm{BB}} \mathbf{F}_{\mathrm{BB}}^H \mathbf{F}_{\mathrm{RF}}^H$. In this paper, we consider the average power constraint, which is given by 
\begin{equation}
  \mathrm{tr}(\mathbf{Q}) = \|\mathbf{F}_{\mathrm{RF}} \mathbf{F}_{\mathrm{BB}}\|_F^2 \le P_t.
\end{equation}
Then, the received signal at user $k, \forall k \in \mathcal{K}_{\chi}, \chi \in \{t,r\},$ is given by 
\begin{align} \label{eqn:narrow_receive_signal}
  y_k[n] = &\mathbf{h}_k^n[n] * \mathbf{x}[n] + n_k[n] \nonumber \\
  = &\mathbf{v}_k^n \mathbf{\Theta}_{\chi} \mathbf{G}^n \mathbf{F}_{\mathrm{RF}} \sum_{k \in \mathcal{K}} \mathbf{f}_{\mathrm{BB},k} s_k[n] + n_k[n] \nonumber \\
  = & \underbrace{\boldsymbol{\theta}_{\chi}^T \mathbf{H}_k^n \mathbf{F}_{\mathrm{RF}} \mathbf{f}_{\mathrm{BB},k} s_k[n]}_{\text{desired signal}} \nonumber \\
  &+  \underbrace{\sum_{i \in \mathcal{K}, i \neq k} \boldsymbol{\theta}_{\chi}^T \mathbf{H}_k^n \mathbf{F}_{\mathrm{RF}} \mathbf{f}_{\mathrm{BB},i} s_i[n]}_{\text{inter-user interference}} + n_k[n],
\end{align}
where $\boldsymbol{\theta}_{\chi} = [\beta_{\chi,1} e^{j \phi_{\chi,1}},\dots,\beta_{\chi,M}^{j \phi_{\chi,M}}]^T, \forall \chi \in \{t,r\},$ denotes the vector of transmission/reflection coefficients of STARS, $\mathbf{H}_k^n = \mathrm{diag}( \mathbf{v}_k^n ) \mathbf{G}^n$ denotes the cascaded channel from the BS to user $k$, and $n_k[n] \sim \mathcal{CN}(0, \sigma_k^2)$ denotes the additive complex Gaussian noise at user $k$. 

\begin{figure}[t!]
  \centering
  \includegraphics[width=0.4\textwidth]{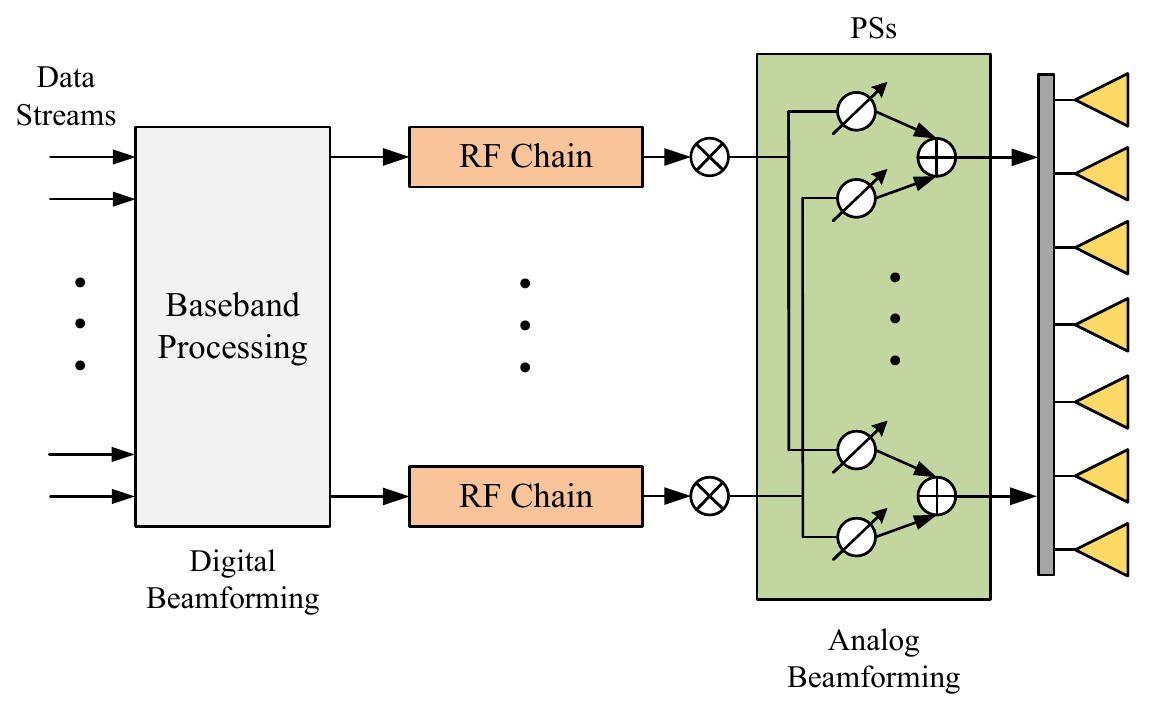}
  \caption{Hybrid beamforming architecture at the BS.}
  \label{fig:narrow_hybrid}
\end{figure}

\subsection{Problem Formulation} \label{sec:narrow_problem_formulation}

Both SE and EE are important performance metrics for communication systems. Specifically, SE is defined as the sum of the achievable rate of communication users. According to \eqref{eqn:narrow_receive_signal},  the achievable rate of user $k, \forall k \in \mathcal{K}_{\chi}, \chi \in \{t,r\},$ is given by
\begin{equation}
  R_k = \log_2 \left( 1 + \frac{| \boldsymbol{\theta}_{\chi}^T \mathbf{H}_k^n \mathbf{F}_{\mathrm{RF}} \mathbf{f}_{\mathrm{BB},k}|^2}{ \sum_{i \in \mathcal{K}, i \neq k} |\boldsymbol{\theta}_{\chi}^T \mathbf{H}_k^n \mathbf{F}_{\mathrm{RF}} \mathbf{f}_{\mathrm{BB},j}|^2 + \sigma_k^2 } \right).
\end{equation} 
The system SE is given by
\begin{equation}
  f_{\mathrm{SE}} = \sum_{k \in \mathcal{K}} R_k.
\end{equation}
Moreover, EE is defined as the ratio of SE to average power consumption. In this work, we exploit a practical rate-dependent power consumption model as follows:
\begin{equation}
  P = \|\mathbf{F}_{\mathrm{RF}} \mathbf{F}_{\mathrm{BB}} \|_F^2 + \xi f_{\mathrm{SE}} + P_c,
\end{equation}
where $\|\mathbf{F}_{\mathrm{RF}} \mathbf{F}_{\mathrm{BB}} \|_F^2$ is the transmit power consumption, $\xi$ denotes the dynamic power consumption per unit data rate incurred by the coding, decoding, and backhaul processes \cite{bjornson2015optimal}, and $P_c$ denotes the rate-independent power consumption as follows. More particularly, $P_c$ can be modeled as follows: 
\begin{equation}
P_c = P_{\mathrm{BS}} + P_{\mathrm{BB}} + N_{\mathrm{RF}} P_{\mathrm{RF}} + N_{\mathrm{RF}} N P_{\mathrm{PS}} + P_{\mathrm{STAR}} + K P_{\mathrm{UE}},
\end{equation}
where $P_{\mathrm{BS}}$, $P_{\mathrm{BB}}$, $P_{\mathrm{RF}}$, $P_{\mathrm{PS}}$, $P_{\mathrm{STAR}}$, and $P_{\mathrm{UE}}$ correspond to the power consumption of the oscillator and circuit at the BS, the baseband processing, each RF chain, each PS, the STARS, and the circuit at each user, respectively. Specifically, $P_{\mathrm{STAR}}$ is defined in equations \eqref{eqn:STAR_power} and \eqref{eqn:STAR_power_coupled} for independent phase-shift and coupled phase-shift STARSs, respectively. Hence, the EE is given by 
\begin{equation}
  f_{\mathrm{EE}} = \frac{f_{\mathrm{SE}}}{P} = \frac{f_{\mathrm{SE}} }{\|\mathbf{F}_{\mathrm{RF}} \mathbf{F}_{\mathrm{BB}} \|_F^2 + \xi f_{\mathrm{SE}}  + P_c}.
\end{equation}

We aim to optimize the hybrid beamformers $\mathbf{F}_{\mathrm{RF}}$ and $\mathbf{F}_{\mathrm{BB}}$ at the BS, along with the transmission and reflection coefficients $\boldsymbol{\theta}_t$ and $\boldsymbol{\theta}_r$ of the STARS, in order to maximize SE and EE. However, maximizing SE and EE requires different design strategies. Specifically, SE maximization aims to use all available power, whereas EE maximization involves balancing SE and power consumption. To investigate the optimization of both SE and EE, we formulate a general optimization problem as follows:
\begin{subequations} \label{problem:narrow_tradeoff_2}
  \begin{align}
      \max_{\mathbf{F}_{\mathrm{RF}}, \mathbf{F}_\mathrm{BB}, \boldsymbol{\theta}_t, \boldsymbol{\theta}_r }  \quad &\frac{f_{\mathrm{SE}}}{w \left(\|\mathbf{F}_{\mathrm{RF}} \mathbf{F}_{\mathrm{BB}} \|_F^2 + \xi f_{\mathrm{SE}}\right) + P_c} \\
      \label{constraint:narrow_1}
      \mathrm{s.t.} \quad & \|\mathbf{F}_{\mathrm{RF}} \mathbf{F}_{\mathrm{BB}} \|_F^2 \le P_t,  \\
      \label{constraint:narrow_2}
      & \boldsymbol{\theta}_{\chi} \in \mathcal{F}, \forall \chi \in \{t,r\}, \\
      \label{constraint:narrow_3}
      & |[\mathbf{F}_{\mathrm{RF}}]_{i,j}| = 1, \forall i, j.
  \end{align}
\end{subequations}
where constraint \eqref{constraint:narrow_1} represents the total power constraint, constraint \eqref{constraint:narrow_2} refers to the transmission and reflection coefficients constraint for either the independent phase-shift STARS or the coupled phase-shift STARS and $\mathcal{F}$ denotes the corresponding feasible set, and constraint \eqref{constraint:narrow_3} is the unit-modulus constraint of PSs. Moreover, a weight factor $w$ is introduced to regulate the rate-dependent power. In particular, the optimization problem \eqref{problem:narrow_tradeoff_2} reduces to the SE maximization problem when $w=0$, and to the EE maximization problem when $w=1$. This problem is challenging to solve due to the non-convex fractional objective function, coupling of optimization variables, non-convex feasible set $\mathcal{F}$, and unit-modulus constraint \eqref{constraint:narrow_3}.

\subsection{Proposed Solution for Independent Phase-shift STARSs} \label{sec:narrow_ideal_solution}
In this subsection, we propose a PDD-based algorithm to solve problem \eqref{problem:narrow_tradeoff_2} for the independent phase shift STARS. In this case, the feasible set $\mathcal{F}$ is given by 
\begin{equation}
  \mathcal{F} = \left\{ \boldsymbol{\theta}_{\chi}, \forall \chi \in \{t,r\} | \beta_{t,m}^2 + \beta_{r,m}^2 = 1, \forall m \right\}
\end{equation}
To solve this problem, we first introduce auxiliary variables $\eta$, $a$, $b$, and $r_k, \forall k \in \mathcal{K}$ to transform \eqref{problem:narrow_tradeoff_2} into a more tractable form. We then prove the following lemma.

\begin{lemma} \label{lemma_1}
  \emph{
    Problem \eqref{problem:narrow_tradeoff_2} is equivalent to the following optimization problem\footnote{With the equivalent transformation in \textbf{Lemma \ref{lemma_1}}, a quality-of-service (QoS) constraint $r_k \ge \overline{R}_k$, where $\overline{R}_k$ denotes the required minimum rate, can be introduced to each user without affecting the convexity of the optimization problem. Therefore, the resulting optimization can also be solved by the algorithms proposed in this work.}:
    \begin{subequations} \label{problem_narrow_transformed}
      \begin{align}
        & \hspace{-0.8cm} \max_{\scriptstyle \mathbf{F}_{\mathrm{RF}}, \mathbf{F}_\mathrm{BB}, \boldsymbol{\theta}_t, \boldsymbol{\theta}_r \atop \scriptstyle \eta, a, b, r_k} \quad \eta \\
          \label{constraint:narrow_transformed_1}
          \hspace{1cm} \mathrm{s.t.} \quad &\eta \le \frac{a^2}{b}, \\
          \label{constraint:narrow_transformed_2}
          & a^2 \le \sum_{k \in \mathcal{K}} r_k, \\
          \label{constraint:narrow_transformed_3}
          & w \big(\|\mathbf{F}_{\mathrm{RF}} \mathbf{F}_{\mathrm{BB}} \|_F^2 + \xi \sum_{k \in \mathcal{K}} r_k \big) + P_c \le b,\\
          \label{constraint:narrow_transformed_4}
          & r_k \le R_k, \forall k \in \mathcal{K}, \\
          & \eqref{constraint:narrow_1}-\eqref{constraint:narrow_3},
      \end{align}
    \end{subequations}
  }
\end{lemma}

\begin{proof}
  According to constraints \eqref{constraint:narrow_transformed_1}--\eqref{constraint:narrow_transformed_3}, it can be readily proved that maximizing $\eta$ is equivalent to maximizing the following objective function:
  \begin{align}
    \tilde{\eta} = &\frac{\sum_{k \in \mathcal{K}} r_k}{w \left(\|\mathbf{F}_{\mathrm{RF}} \mathbf{F}_{\mathrm{BB}} \|_F^2 + \xi \sum_{k \in \mathcal{K}} r_k \right) + P_c} \nonumber \\
    = & \left( \frac{w \|\mathbf{F}_{\mathrm{RF}} \mathbf{F}_{\mathrm{BB}} \|_F^2 + P_c}{\sum_{k \in \mathcal{K}} r_k} + w \xi \right)^{-1}.
  \end{align} 
  It can be seen that $\tilde{\eta}$ increases as $\sum_{k \in \mathcal{K}} r_k$ increase. Then, according to constraint \eqref{constraint:narrow_transformed_4}, we have the maximum value of $\sum_{k \in \mathcal{K}} r_k$ is $f_{\mathrm{SE}} = \sum_{k \in \mathcal{K}} R_k$. As such, the maximum value of $\tilde{\eta}$ must be the maximum value of the following function
  \begin{equation}
    \bar{\tilde{\eta}} = \frac{f_{\mathrm{SE}}}{w \left(\|\mathbf{F}_{\mathrm{RF}} \mathbf{F}_{\mathrm{BB}} \|_F^2 + \xi f_{\mathrm{SE}}\right) + P_c},
  \end{equation}   
  which is the objective function of problem \eqref{problem:narrow_tradeoff_2}. Thus, the proof is completed.
\end{proof}

For the equivalent optimization problem in \textbf{Lemma \ref{lemma_1}}, the objective function is only related to the auxiliary variable $\eta$. However, the coupling between the remaining optimization variables still makes it difficult to solve the new optimization problem. To address this issue, we introduce additional auxiliary variables $\mathbf{F} = [\mathbf{f}_1,\dots,\mathbf{f}_K] = \mathbf{F}_{\mathrm{RF}} \mathbf{F}_{\mathrm{BB}}$ and $\mathbf{p}_k = \boldsymbol{\theta}_{\chi}^T \mathbf{H}_k^n \mathbf{F}, \forall k \in \mathcal{K}_{\chi}, \chi \in \{t,r\}$. In this case, the achievable rate of user $k$ can be rewritten as a function of $\mathbf{p}_k$, i.e., 
\begin{equation}
  R_k(\mathbf{p}_k) = \log_2 \left( 1 + \frac{|p_{k,k}|^2}{ \sum_{i \in \mathcal{K}, i \neq k} |p_{k,i}|^2 + \sigma_k^2 }\right),
\end{equation} 
where $p_{k,i}$ denotes the $i$-th entry of $\mathbf{p}_k$.   
Then, problem \eqref{problem_narrow_transformed} can be further transformed into the following equivalent optimization problem:
\begin{subequations} \label{problem:narrow_equality}
  \begin{align}
      \max_{\scriptstyle \mathbf{F}_{\mathrm{RF}}, \mathbf{F}_\mathrm{BB}, \mathbf{F}, \mathbf{p}_k, \atop \scriptstyle \boldsymbol{\theta}_t, \boldsymbol{\theta}_r, \eta, a, b, r_k} \quad &\eta \\
      \label{constraint:narrow_transmit_power_transformed}
      \mathrm{s.t.} \quad & \|\mathbf{F}\|_F^2 \le P_t, \\
      \label{constraint:narrow_power_transformed}
      & w \big(\|\mathbf{F} \|_F^2 + \xi \sum_{k \in \mathcal{K}} r_k \big) + P_c \le b,\\
      \label{constraint:narrow_rate_transformed}
      & r_k \le R_k(\mathbf{p}_k), \forall k \in \mathcal{K}, \\
      \label{constraint:narrow_equal_1}
      & \mathbf{F} = \mathbf{F}_{\mathrm{RF}} \mathbf{F}_{\mathrm{BB}}, \\
      \label{constraint:narrow_equal_2}
      & \mathbf{p}_k = \boldsymbol{\theta}_{\chi}^T \mathbf{H}_k^n \mathbf{F}, \forall k \in \mathcal{K}_{\chi}, \chi \in \{t,r\},\\
      & \eqref{constraint:narrow_2}, \eqref{constraint:narrow_3}, \eqref{constraint:narrow_transformed_1}, \eqref{constraint:narrow_transformed_2}.
  \end{align}
\end{subequations}
In this optimization, all couplings have been transformed into equality constraints \eqref{constraint:narrow_equal_1} and \eqref{constraint:narrow_equal_2}, which motivates us to apply the PDD framework \cite{shi2020penalty} to solve it. The PDD framework relies on constructing an augmented Lagrangian (AL) problem for the original optimization problem. By introducing the dual variables $\mathbf{\Psi} \in \mathbb{C}^{N \times K}$ and $\boldsymbol{\lambda}_k \in \mathbb{C}^{1 \times K}, \forall k \in \mathcal{K},$ for the equality constraints \eqref{constraint:narrow_equal_1} and \eqref{constraint:narrow_equal_2}, respectively, the following AL problem of problem \eqref{problem:narrow_equality} can be formulated:
\begin{subequations} \label{problem:narrow_AL}
  \begin{align}
      \max_{\scriptstyle \mathbf{F}_{\mathrm{RF}}, \mathbf{F}_\mathrm{BB}, \mathbf{F}, \mathbf{p}_k, \atop \scriptstyle \boldsymbol{\theta}_t, \boldsymbol{\theta}_r, \eta, a, b, r_k} \quad &\eta
      - \frac{1}{2 \rho} \|\mathbf{F} - \mathbf{F}_{\mathrm{RF}} \mathbf{F}_{\mathrm{BB}} + \rho \mathbf{\Psi} \|_F^2 \nonumber \\
      - \sum_{\chi \in \{t,r\}} &\sum_{k \in \mathcal{K}_{\chi}} \frac{1}{2\rho} \| \mathbf{p}_k - \boldsymbol{\theta}_{\chi}^T \mathbf{H}_k^n \mathbf{F} + \rho \boldsymbol{\lambda}_k\|^2  \\
      &\hspace{-2cm} \mathrm{s.t.} \quad \eqref{constraint:narrow_2}, \eqref{constraint:narrow_3}, \eqref{constraint:narrow_transformed_1}, \eqref{constraint:narrow_transformed_2}, \eqref{constraint:narrow_transmit_power_transformed} - \eqref{constraint:narrow_rate_transformed},
  \end{align}
\end{subequations}
where $\rho \ge 0$ is the penalty factor. The PDD framework follows a double-loop structure, where the inner loop solves the AL problem \eqref{problem:narrow_AL} and the outer loop updates the dual variables and the penalty factor based on the results obtained in the inner loop. Note that the AL problem \eqref{problem:narrow_AL} is separable, which motivates us to solve it through block coordinate descent (BCD). Specifically, we divide the optimization variables into four blocks, namely $\{\mathbf{F}, \mathbf{p}_k, \eta, a, b, r_k\}$, $\{\boldsymbol{\theta}_t, \boldsymbol{\theta}_r\}$, $\mathbf{F}_{\mathrm{RF}}$, and $\mathbf{F}_{\mathrm{BB}}$. In each iteration of the inner loop, each block is updated sequentially by fixing the other blocks. The details for developing the PDD-based algorithm are given as follows.

\subsubsection{Subproblem With Respect to $\{\mathbf{F}, \mathbf{p}_k, \eta, a, b, r_k\}$}
By fixing the other blocks, the subproblem with respect to $\{\mathbf{F}, \mathbf{p}_k, \eta, a, b, r_k\}$ is given by
\begin{subequations} \label{problem:narrow_subproblem_1}
  \begin{align}
      \max_{\scriptstyle \mathbf{F}, \mathbf{p}_k, \atop \scriptstyle \eta, a, b, r_k} \quad &\eta
      - \frac{1}{2 \rho} \|\mathbf{F} - \mathbf{F}_{\mathrm{RF}} \mathbf{F}_{\mathrm{BB}} + \rho \mathbf{\Psi} \|_F^2 \nonumber \\
      &- \sum_{\chi \in \{t,r\}} \sum_{k \in \mathcal{K}_{\chi}} \frac{1}{2\rho} \| \mathbf{p}_k - \boldsymbol{\theta}_{\chi}^T \mathbf{H}_k^n \mathbf{F} + \rho \boldsymbol{\lambda}_k\|^2  \\
      \mathrm{s.t.} \quad & \eqref{constraint:narrow_transformed_1}, \eqref{constraint:narrow_transformed_2}, \eqref{constraint:narrow_transmit_power_transformed} - \eqref{constraint:narrow_rate_transformed}.
  \end{align}
\end{subequations}
The non-convexity of this problem only lies in constraints \eqref{constraint:narrow_transformed_1} and \eqref{constraint:narrow_rate_transformed}. We observe that the right-hand side of the constraint \eqref{constraint:narrow_transformed_1}, i.e., $a^2/b$, is a quadratic-over-linear function, which is jointly convex for $a$ and $b$. This motivates us to exploit successive convex approximation (SCA) to approximate it. Let $\breve{a}$ and $\breve{b}$ denote the results of $a$ and $b$ obtained in the previous iteration of the inner loop, respectively. The following convex lower bound of $a^2/b$ can be obtained via the first-order Taylor expansion:
\begin{equation} \label{eqn:narrow_sca}
  \frac{a^2}{b} \ge \frac{2 \breve{a}}{\breve{b}} a - \left( \frac{\breve{a}}{\breve{b}} \right)^2 b \triangleq \breve{\varpi}(a,b).
\end{equation}
Now, we show that the constraint \eqref{constraint:narrow_rate_transformed} can be approximated in a similar way. By defining $\gamma_k$ as the signal-to-interference-plus-noise ratio (SINR) term in $\mathbf{R}_k(\mathbf{p}_k)$, the constraint \eqref{constraint:narrow_rate_transformed} can be reformulated as 
\begin{equation}
  2^{r_k} - 1 \le \frac{|p_{k,k}|^2}{ \sum_{i \in \mathcal{K}, i \neq k} |p_{k,i}|^2 + \sigma_k^2} =  \gamma_k, \forall k \in \mathcal{K}.
\end{equation}
The expression of $\gamma_k$ is also in a quadratic-over-linear form. Let $\breve{\mathbf{p}}_k$ denote the result of $\mathbf{p}_k$ obtained in the previous iteration of the inner loop. The following convex lower bound of $\gamma_k$ can be obtained based on the results in \eqref{eqn:narrow_sca}:
\begin{align}
  \gamma_k \ge &\frac{2 \mathrm{Re} \{ \breve{p}_{k,k}^* p_{k,k} \} }{I_k(\breve{\mathbf{p}}_k)} - \left| \frac{\breve{p}_{k,k}}{I_k(\breve{\mathbf{p}}_k)} \right|^2 I_k(\mathbf{p}_k) \triangleq \breve{\gamma}_k(\mathbf{p}_k),
\end{align}
where $I_k(\mathbf{p}_k) =\sum_{i \in \mathcal{K}, i \neq k} |p_{k,i}|^2 + \sigma_k^2$ and $\breve{p}_{k,i}$ denotes the $i$-th entry of $\breve{\mathbf{p}}_k$. As a result, problem \eqref{problem:narrow_subproblem_1} can be approximated by the following optimization problem:
\begin{subequations} \label{problem:narrow_subproblem_1_1}
  \begin{align}
      \max_{\scriptstyle \mathbf{F}, \mathbf{p}_k, \atop \scriptstyle \eta, a, b, r_k} \quad &\eta
      - \frac{1}{2 \rho} \|\mathbf{F} - \mathbf{F}_{\mathrm{RF}} \mathbf{F}_{\mathrm{BB}} + \rho \mathbf{\Psi} \|_F^2 \nonumber \\
      &- \sum_{\chi \in \{t,r\}} \sum_{k \in \mathcal{K}_{\chi}} \frac{1}{2\rho} \| \mathbf{p}_k - \boldsymbol{\theta}_{\chi}^T \mathbf{H}_k^n \mathbf{F} + \rho \boldsymbol{\lambda}_k\|^2  \\
      \mathrm{s.t.} \quad & \eta \le \breve{\varpi}(a,b), \\
      & 2^{r_k} -1 \le \breve{\gamma}_k(\mathbf{p}_k), \forall k \in \mathcal{K}, \\
      & \eqref{constraint:narrow_transformed_2}, \eqref{constraint:narrow_transmit_power_transformed}, \eqref{constraint:narrow_power_transformed}.
  \end{align}
\end{subequations}
The above optimization problem is convex and can be effectively solved by the standard interior-point algorithm. 

\subsubsection{Subproblem With Respect to $\{ \boldsymbol{\theta}_t, \boldsymbol{\theta}_r \}$}
The block $\{ \boldsymbol{\theta}_t, \boldsymbol{\theta}_r \}$ only appears in the last term of the objective function and the constraint \eqref{constraint:narrow_2} of the AL problem. Thus, the corresponding subproblem is given by   
\begin{subequations} \label{problem:narrow_subproblem_2}
  \begin{align}
      \min_{\boldsymbol{\theta}_t, \boldsymbol{\theta}_r} \quad &\sum_{\chi \in \{t,r\}} \sum_{k \in \mathcal{K}_{\chi}} \| \mathbf{p}_k - \boldsymbol{\theta}_{\chi}^T \mathbf{H}_k^n \mathbf{F} + \rho \boldsymbol{\lambda}_k\|^2 \\
      \mathrm{s.t.} \quad & \beta_{t,m}^2 + \beta_{r,m}^2 = 1, \forall m,
  \end{align}
\end{subequations}
which is non-convex due to the quadratic equality constraint. To address this issue, we propose a low-complexity element-wise algorithm to solve it, where one entry of $\boldsymbol{\theta}_{\chi}, \forall \chi \in \{t,r\},$ is optimized at each iteration while fixing the others. In particular, recall that $\beta_{\chi,m} \in [0,1]$ and $\phi_{\chi,m} \in [0, 2\pi]$ denote the amplitude and phase of the $m$-th entry of $\boldsymbol{\theta}_{\chi}$, respectively. Then, the subproblem with respect $\beta_{\chi,m}$ and $\phi_{\chi,m}, \forall \chi \in \{t,r\},$ is given in the following lemma.

\begin{lemma} \label{lemma_2}
  \emph{By fixing the other entries of $\boldsymbol{\theta}_{\chi}, \forall \chi \in \{t,r\},$ the optimization problem with respect to $\beta_{\chi,m}$ and $\phi_{\chi,m}, \forall \chi \in \{t,r\},$ is given by
  \begin{subequations} \label{problem:narrow_theta_1}
    \begin{align}
        \min_{ \scriptstyle \beta_{t,m}, \beta_{r,m} \atop \scriptstyle \phi_{t,m}, \phi_{r,m}} \quad &\sum_{\chi \in \{t,r\}} c_{\chi,m} \beta_{\chi,m}^2 -  2 \beta_{\chi,m} \mathrm{Re} \{ d_{\chi,m}^* e^{j \phi_{\chi,m}} \} \\
        \mathrm{s.t.} \quad & \beta_{t,m}, \beta_{r,m} \in [0,1], \beta_{t,m}^2 + \beta_{r,m}^2 = 1,
    \end{align}
  \end{subequations}
  where $c_{\chi,m}$ and $d_{\chi,m}$ are given in \eqref{aux_theta}.}  
\end{lemma}

\begin{proof}
  Please refer to Appendix A.
\end{proof}

For problem \eqref{problem:narrow_theta_1}, it is not difficult to show that the optimal value of $\phi_{\chi,m}$ is given by 
\begin{equation} \label{optimal_phi}
  \phi_{\chi,m}^\star = \angle d_{\chi,m}, \forall \chi \in \{t,r\},
\end{equation}
which is independent to the value of $\beta_{\chi,m}$. Therefore, it only remains to find the optimal value of $\beta_{\chi,m}$. To this end, we define an auxiliary variable $\vartheta_m \in [0, \frac{\pi}{2}]$ such that $\beta_{t,m} = \sin \vartheta_m$ and $\beta_{r,m} = \cos \vartheta_m$. In this case, the constraints on $\beta_{t,m}$ and $\beta_{r,m}$ are automatically satisfied. Then, by substituting \eqref{optimal_phi}, the optimization problem \eqref{problem:narrow_theta_1} can be reformulated as follows:
\begin{subequations} \label{problem:narrow_theta_2}
  \begin{align}
      \min_{\vartheta_m} \quad &c_{t,m} \sin^2 \vartheta_m + c_{r,m} \cos^2 \vartheta_m \nonumber \\
      & -  2 |d_{t,m}|\sin \vartheta_m -  2 |d_{r,m}|\cos \vartheta_m\\
      \mathrm{s.t.} \quad & \vartheta_m \in [0, \frac{\pi}{2}].
  \end{align}
\end{subequations}
The above optimization problem is to find the minimum of a single-variable function on a fixed interval. Such a problem can be effectively solved by some standard methods, such as golden-section search. Let $\theta_m^\star$ denote the optimal solution to problem \eqref{problem:narrow_theta_2}. Then, the optimal $\beta_{t,m}$ and $\beta_{r,m}$ are given by 
\begin{equation} \label{optimal_beta}
  \beta_{t,m}^\star = \sin \vartheta_m^\star, \beta_{r,m}^\star = \cos \vartheta_m^\star.
\end{equation} 
Therefore, problem \eqref{problem:narrow_subproblem_2} can be solved by exploiting \textbf{Algorithm \ref{alg:narrow_theta}}.

\begin{algorithm}[tb]
  \caption{Element-wise algorithm for solving \eqref{problem:narrow_subproblem_2}.}
  \label{alg:narrow_theta}
  \begin{algorithmic}[1]
      \STATE{initialize $\boldsymbol{\theta}_{\chi}, \forall \chi \in \{t,r\}$.}
      \REPEAT
      \FOR{ $m \in \{1,\dots,M\}$  }
        \STATE{calculate $c_{\chi,m}$ and $d_{\chi,m}$, according to \eqref{aux_theta}.  }
        \STATE{calculate $\phi_{\chi,m}^\star,$ according to \eqref{optimal_phi}. }
        \STATE{calculate $\beta_{\chi,m}^\star,$ according to \eqref{optimal_beta}. }
        \STATE{update $[\boldsymbol{\theta}_{\chi}]_m$ as $\beta_{\chi,m}^\star e^{j \phi_{\chi,m}^\star}$.  }
      \ENDFOR
      \UNTIL{the fractional reduction of the objective value falls below a predefined threshold.}
  \end{algorithmic}
\end{algorithm}

\subsubsection{Subproblem With Respect to $\mathbf{F}_{\mathrm{RF}}$}
The block $\mathbf{F}_{\mathrm{RF}}$ only contributes to the second term in the objective function and the unit-modulus constraint \eqref{constraint:narrow_3} of the AL problem, leading to the following subproblem:
\begin{subequations} \label{problem:narrow_subproblem_3}
  \begin{align}
      \min_{\mathbf{F}_{\mathrm{RF}}} \quad & \| \mathbf{F} - \mathbf{F}_{\mathrm{RF}} \mathbf{F}_{\mathrm{BB}} + \rho \mathbf{\Psi} \|_F^2\\
      \mathrm{s.t.} \quad & |[\mathbf{F}_{\mathrm{RF}}]_{i,j}| = 1, \forall i, j.
  \end{align}
\end{subequations}
The above subproblem can be transformed into the following equivalent form: 
\begin{subequations} \label{problem:narrow_subproblem_3_1}
  \begin{align}
      \min_{\mathbf{F}_{\mathrm{RF}}} \quad & \mathrm{tr}( \mathbf{F}_{\mathrm{RF}}^H \mathbf{F}_{\mathrm{RF}} \mathbf{A}) - 2 \mathrm{Re}\{ \mathrm{tr}( \mathbf{F}_{\mathrm{RF}}^H \mathbf{B} ) \}  \\
      \mathrm{s.t.} \quad & |[\mathbf{F}_{\mathrm{RF}}]_{i,j}| = 1, \forall i, j,
  \end{align}
\end{subequations}
where $\mathbf{A} = \mathbf{F}_{\mathrm{BB}} \mathbf{F}_{\mathrm{BB}}^H$ and $\mathbf{B} = (\mathbf{F} + \rho \mathbf{\Psi}) \mathbf{F}_{\mathrm{BB}}^H$. This problem is non-convex due to the unit-modulus constraint. Similarly, this problem can be solved by a low-complexity element-wise algorithm. Following the same path in Appendix A, the subproblem with respect to the $(i,j)$-th entry of $[\mathbf{F}_{\mathrm{RF}}]_{i,j}$ can be expressed as
\begin{subequations} \label{problem:narrow_subproblem_3_2}
  \begin{align}
    \min_{[\mathbf{F}_{\mathrm{RF}}]_{i,j}} \quad & p_{i,j} |[\mathbf{F}_{\mathrm{RF}}]_{i,j}|^2 - 2 \mathrm{Re}\{ q_{i,j}^* [\mathbf{F}_{\mathrm{RF}}]_{i,j} \}, \\
    \mathrm{s.t.} \quad & |[\mathbf{F}_{\mathrm{RF}}]_{i,j}| = 1,
\end{align}
\end{subequations}
where $p_{i,j}$ is some real number and $q_{i,j}$ is given by 
\begin{equation} \label{eqn:narrow_q_ij}
  q_{i,j} = [\mathbf{F}_{\mathrm{RF}}]_{i,j} [\mathbf{A}]_{j,j} - [\mathbf{F}_{\mathrm{RF}} \mathbf{A}]_{i,j} + [\mathbf{B}]_{i,j}.
\end{equation}
Given that $|[\mathbf{F}_{\mathrm{RF}}]_{i,j}| = 1$, problem \eqref{problem:narrow_subproblem_3_2} is equivalent to maximizing $\mathrm{Re}\{ q_{i,j}^* [\mathbf{F}_{\mathrm{RF}}]_{i,j} \}$ subject to the unit-modulus constant. It can be readily obtained that the optimal solution is 
\begin{equation}
  [\mathbf{F}_{\mathrm{RF}}]_{i,j} = \frac{q_{i,j}}{|q_{i,j}|}.
\end{equation}
Therefore, problem \eqref{problem:narrow_subproblem_3_1} can be efficiently solved using the algorithm described in \textbf{Algorithm \ref{alg:narrow_fixed_point}}.

\begin{algorithm}[tb]
  \caption{Element-wise algorithm for solving \eqref{problem:narrow_subproblem_3_1}.}
  \label{alg:narrow_fixed_point}
  \begin{algorithmic}[1]
      \STATE{initialize $\mathbf{F}_{\mathrm{RF}}$.}
      \REPEAT
      \FOR{ $(i,j) \in \{1,\dots,N\} \times \{1,\dots,N_{\mathrm{RF}}\}$  }
        \STATE{calculate $q_{i,j}$ according to \eqref{eqn:narrow_q_ij}.  }
        \STATE{update $[\mathbf{F}_{\mathrm{RF}}]_{i,j}$ as $q_{i,j} / |q_{i,j}|$. }
      \ENDFOR
      \UNTIL{the fractional reduction of the objective value falls below a predefined threshold.}
  \end{algorithmic}
\end{algorithm}

\subsubsection{Subproblem With Respect to $\mathbf{F}_{\mathrm{BB}}$}
The subproblem with respect to $\mathbf{F}_{\mathrm{BB}}$ is given by
\begin{align}
  \label{problem:narrow_subproblem_4}
  \min_{\mathbf{F}_{\mathrm{BB}}} \quad & \| \mathbf{F} - \mathbf{F}_{\mathrm{RF}} \mathbf{F}_{\mathrm{BB}} + \rho \mathbf{\Psi} \|_F^2.
\end{align}
The above problem is an unconstrained convex optimization problem. Thus, the optimal solution can be obtained by the first-order optimality condition and is given by 
\begin{equation} \label{eqn:narrow_subproblem_4_solution}
  \mathbf{F}_{\mathrm{BB}}^\star = (\mathbf{F}_{\mathrm{RF}}^H \mathbf{F}_{\mathrm{RF}})^{-1} \mathbf{F}_{\mathrm{RF}}^H (\mathbf{F} + \rho \mathbf{\Psi}).
\end{equation}

\subsubsection{Update Dual Variables and Penalty Factor}
With the proposed solutions for the above four subproblems, the AL problem \eqref{problem:narrow_AL} can be efficiently solved by iteratively updating the blocks $\{\mathbf{F}, \mathbf{p}_k, \eta, a, b, r_k\}$, $\{\boldsymbol{\theta}_t, \boldsymbol{\theta}_r\}$, $\mathbf{F}_{\mathrm{RF}}$, and $\mathbf{F}_{\mathrm{BB}}$ within the inner loop of the PDD framework.  
In the outer loop, the dual variables and the penalty factor are updated according to the following policy \cite{shi2020penalty}. Firstly, we define the constraint violation function as 
\begin{equation}
  h = \max \left\{\begin{array}{c}
    \|\mathbf{F} - \mathbf{F}_{\mathrm{RF}} \mathbf{F}_{\mathrm{BB}} \|_{\infty},\\
    \max_{k,\chi}  \| \mathbf{p}_k - \boldsymbol{\theta}_{\chi}^T \mathbf{H}_k^n \mathbf{F}\|_{\infty}
    \end{array}
    \right\}.
\end{equation}
If the constraint violation function $h$ is smaller than the predefined threshold $\varepsilon$ at the $n$-th iteration of the outer loop, the penalty factor keeps unchanged, while the dual variables are updated based on gradient descent as follows:
  \begin{subequations} \label{eqn:narrow_dual_update}
    \begin{align}
      &\mathbf{\Psi} \leftarrow \mathbf{\Psi} + \frac{1}{\rho} (\mathbf{F} - \mathbf{F}_{\mathrm{RF}} \mathbf{F}_{\mathrm{BB}}), \\
      & \boldsymbol{\lambda}_k \leftarrow \boldsymbol{\lambda}_k + \frac{1}{\rho} (\mathbf{p}_k - \boldsymbol{\theta}_{\chi}^T \mathbf{H}_k^n \mathbf{F}), \forall k \in \mathcal{K}_{\chi}, \chi \in \{t,r\}.
    \end{align} 
  \end{subequations}
If the constraint violation function is larger than the predefined threshold, the dual variables keep unchanged, and the penalty factor is updated by $\rho \leftarrow \kappa \rho$, where $0 <\kappa < 1$ is a reduction factor. The proposed PDD-based algorithm for problem \eqref{problem:narrow_equality} is summarized in \textbf{Algorithm \ref{alg:narrow_PDD}}. 

\subsubsection{Initialization, Convergence, and Complexity}
The initial optimization variables of \textbf{Algorithm \ref{alg:narrow_PDD}} are generated as follows. For the BS, the analog beamformer $\mathbf{F}_{\mathrm{RF}}$ is initialized based on the knowledge of the physical directions of the channel in \eqref{eqn:channel_path}. More specifically, the $N_{\mathrm{RF}}$ strongest paths are firstly selected from the total $L$ paths between the BS and the STARS. The analog beamformer is then initialized to generate directional beams toward the physical directions of these paths. For example, for the $n$-th strongest path with the physical direction $\varphi_n$, the $n$-th column $\mathbf{f}_{\mathrm{RF},n}^{\text{init}}$ of the initial analog beamformer is given by    
\begin{equation}
  \mathbf{f}_{\mathrm{RF},n}^{\text{init}} = \mathbf{b}(f, \varphi_n).
\end{equation}
Then, the digital beamformer $\mathbf{F}_{\mathrm{BB}}$ and the STARS coefficients $\{\boldsymbol{\theta}_t, \boldsymbol{\theta}_r\}$ are randomly initialized such that the constraint \eqref{constraint:narrow_2} and the constraint \eqref{eqn:ideal_STAR} are satisfied, respectively. The auxiliary variables are initialized such that the corresponding equality constraints are satisfied. 

In \textbf{Algorithm \ref{alg:narrow_PDD}}, for any given dual variables $\mathbf{\Psi}$ and $\boldsymbol{\lambda}_k$ and the penalty factor $\rho$, the AL problem \eqref{problem:narrow_AL} is solved by applying BCD with non-decreasing objective value over iterations in the inner loop. Moreover, the objective value is upper-bounded because of the maximum power constraint. Thus, according to the analysis in \cite{shi2020penalty}, the proposed algorithm is guaranteed to converge to a stationary point of problem \eqref{problem:narrow_equality}, which is also a stationary point of the original problem \eqref{problem:narrow_tradeoff_2}.

The complexity of \textbf{Algorithm \ref{alg:narrow_PDD}} primarily stems from the BCD iterations in the inner loop. Specifically, each BCD iteration involves solving problem \eqref{problem:narrow_subproblem_1_1}, the complexity of which is dominated by the second-order cone (SOC) constraints. This problem has $N_o = NK + K^2 + K + 3$ optimization variables, $K$ SOC constraints with a dimension of $K-1$, and two SOC constraints with a dimension of $NK$. Therefore, the complexity of solving problem \eqref{problem:narrow_subproblem_1_1} is $\mathcal{O}(N_o^3 + N_o^2 (K(K-1)^2 + 2N^2K^2))$ \cite{ben2001lectures}, where $\mathcal{O}(\cdot)$ is the big-O notation. Problems \eqref{problem:narrow_subproblem_2} and \eqref{problem:narrow_subproblem_3} are solved in an element-wise manner, which has a complexity of $\mathcal{O}(1)$ of each step. 
Finally, calculating the closed-form $\mathbf{F}_{\mathrm{BB}}$ according to \eqref{eqn:narrow_subproblem_4_solution} has a main complexity arising from the matrix inversion operation, which is $\mathcal{O}(N_{\mathrm{RF}}^3)$.

\begin{algorithm}[tb]
  \caption{PDD-based algorithm for solving \eqref{problem:narrow_equality}.}
  \label{alg:narrow_PDD}
  \begin{algorithmic}[1]
      \STATE{initialize the optimization variables, and set $0<c<1$.}
      \REPEAT
      \REPEAT
        \STATE{update $\{\mathbf{F}, \mathbf{p}_k, \eta, a, b, r_k\}$ by solving problem \eqref{problem:narrow_subproblem_1_1}. }
        \STATE{update $\{\boldsymbol{\theta}_t, \boldsymbol{\theta}_r\}$ by \textbf{Algorithm \ref{alg:narrow_theta}}.  }
        \STATE{update $\mathbf{F}_{\mathrm{RF}}$ by \textbf{Algorithm \ref{alg:narrow_fixed_point}}. }
        \STATE{update $\mathbf{F}_{\mathrm{BB}}$ by \eqref{eqn:narrow_subproblem_4_solution}. }
      \UNTIL{convergence.}
      \IF{$h \le \varepsilon$}
          \STATE{update the dual variables $\mathbf{\Psi}$ and $\{\boldsymbol{\lambda}_k\}_{k \in \mathcal{K}}$ by \eqref{eqn:narrow_dual_update}.}
      \ELSE
          \STATE{update the penalty factor as $\rho \leftarrow c \rho$.}
      \ENDIF
      \STATE{set $\varepsilon = 0.9 h$.}
      \UNTIL{$h$ falls below a predefined threshold.}
  \end{algorithmic}
\end{algorithm}

\subsection{Proposed Solution for Coupled Phase-shift STARSs} \label{sec:narrow_coupled_solution}
In this subsection, we extend the proposed PDD-based algorithm to the case of coupled phase-shift STARSs, where the feasible set $\mathcal{F}$ becomes
\begin{equation}
    \mathcal{F} = \left\{ \boldsymbol{\theta}_{\chi}, \forall \chi \in \{t,r\} \Big|  \begin{array}{l}
      \beta_{t,m}^2 + \beta_{r,m}^2 = 1, \forall m \\
      \cos(\phi_{t,m} - \phi_{r,m}) = 0, \forall m
    \end{array}  \right\}.
\end{equation} 
To address the additional coupled phase-shift constraint, we further introduce an equality constraint as follows \cite{wang2022coupled}:
\begin{equation}
  \boldsymbol{\vartheta}_{\chi} = \boldsymbol{\theta}_{\chi}, \forall \chi \in \{t,r\},
\end{equation}
where $\boldsymbol{\vartheta}_{\chi} = [\tilde{\beta}_{i,1}e^{j \tilde{\phi}_{i,1}} ,\dots,\tilde{\beta}_{i,M} e^{j \tilde{\phi}_{i,M}}]^T, \forall \chi \in \{t,r\}, $ is the auxiliary variable. Based on \textbf{Lemma \ref{lemma_1}}, the resulting optimization problem can be transformed into the following equivalent form:
\begin{subequations} \label{problem:narrow_coupled}
  \begin{align}
      & \hspace{-2cm} \max_{\scriptstyle \mathbf{F}_{\mathrm{RF}}, \mathbf{F}_\mathrm{BB}, \mathbf{F}, \mathbf{p}_k, \atop \scriptstyle \boldsymbol{\theta}_t, \boldsymbol{\theta}_r, \boldsymbol{\vartheta}_t, \boldsymbol{\vartheta}_r, \eta, a, b, r_k} \quad \eta \\
      \label{constraint:narrow_equality_STAR}
      \mathrm{s.t.} \quad & \boldsymbol{\vartheta}_{\chi} = \boldsymbol{\theta}_{\chi}, \forall \chi \in \{t,r\}, \\
      \label{constraint:narrow_coupled_1}
      & \tilde{\beta}_{t,m}^2 + \tilde{\beta}_{r,m}^2 = 1, \forall m \in \mathcal{M}, \\
      \label{constraint:narrow_coupled_2}
      & \cos(\tilde{\phi}_{t,m} - \tilde{\phi}_{r,m}) = 0, \forall m \in \mathcal{M}, \\
      & \eqref{constraint:narrow_3}, \eqref{constraint:narrow_transformed_1}, \eqref{constraint:narrow_transformed_2}, \eqref{constraint:narrow_power_transformed} - \eqref{constraint:narrow_equal_2}.
  \end{align}
\end{subequations}
In this problem, the complicated constraint in the feasible set $\mathcal{F}$ is transferred to the auxiliary variables $\boldsymbol{\vartheta}_{\chi}$ and the optimization variables $\boldsymbol{\theta}_{\chi}$ is only subject to the equality constraint \eqref{constraint:narrow_equality_STAR}. Therefore, by defining the dual variables $\boldsymbol{\mu}_i \in \mathbb{C}^{M \times 1}, \forall \chi \in \{t,r\}, $ for the equality constraint \eqref{constraint:narrow_equality_STAR}, we can formulate the following AL problem:
\begin{subequations}
  \begin{align}
      \max_{\scriptstyle \mathbf{F}_{\mathrm{RF}}, \mathbf{F}_\mathrm{BB}, \mathbf{F}, \mathbf{p}_k, \atop \scriptstyle \boldsymbol{\theta}_t, \boldsymbol{\theta}_r, \boldsymbol{\vartheta}_t, \boldsymbol{\vartheta}_r, \eta, a, b, r_k} & \quad \eta - \frac{1}{2 \rho} \|\mathbf{F} - \mathbf{F}_{\mathrm{RF}} \mathbf{F}_{\mathrm{BB}} + \rho \mathbf{\Psi} \|_F^2 \nonumber \\
      - & \sum_{\chi \in \{t,r\}} \sum_{k \in \mathcal{K}_{\chi}} \frac{1}{2\rho} \| \mathbf{p}_k - \boldsymbol{\theta}_{\chi}^T \mathbf{H}_k^n \mathbf{F} + \rho \boldsymbol{\lambda}_k\|^2 \nonumber \\
      - & \sum_{\chi \in \{t,r\}} \frac{1}{2\rho} \| \boldsymbol{\vartheta}_{\chi} - \boldsymbol{\theta}_{\chi} + \rho \boldsymbol{\mu}_i\|^2 \\
      \mathrm{s.t.} \quad  \eqref{constraint:narrow_3}, \eqref{constraint:narrow_transformed_1}, & \eqref{constraint:narrow_transformed_2}, \eqref{constraint:narrow_power_transformed}, \eqref{constraint:narrow_rate_transformed}, \eqref{constraint:narrow_coupled_1}, \eqref{constraint:narrow_coupled_2},
  \end{align}
\end{subequations}
In a similar manner to the independent phase-shift STARS case, the AL problem can be efficiently solved through BCD by dividing the optimization variables into five blocks, which include $\{\mathbf{F}, \mathbf{p}_k, \eta, a, b, r_k\}$, $\{\boldsymbol{\theta}_t, \boldsymbol{\theta}_r\}$, $\{\boldsymbol{\vartheta}_t, \boldsymbol{\vartheta}_r\}$, $\mathbf{F}_{\mathrm{RF}}$, and $\mathbf{F}_{\mathrm{BB}}$. Each of these blocks can be updated using the same method as that employed in the independent phase-shift STARS case, with the exception of $\{\boldsymbol{\theta}_t, \boldsymbol{\theta}_r\}$ and $\{\boldsymbol{\vartheta}_t, \boldsymbol{\vartheta}_r\}$, which require new optimization approaches due to differences in their respective objective functions and constraints. Therefore, we focus on solving the subproblems with respect to these two blocks in the following.

\subsubsection{Subproblem With Respect to $\{\boldsymbol{\theta}_t, \boldsymbol{\theta}_r\}$}
The block $\{\boldsymbol{\theta}_t, \boldsymbol{\theta}_r\}$ appears in the second and third terms of the objective function of the AL problem. The corresponding subproblem is thus given by 
\begin{align}
    \min_{\boldsymbol{\theta}_t, \boldsymbol{\theta}_r} \quad \sum_{\chi \in \{t,r\}} &\sum_{k \in \mathcal{K}_{\chi}} \| \mathbf{p}_k - \boldsymbol{\theta}_{\chi}^T \mathbf{H}_k^n \mathbf{F} + \rho \boldsymbol{\lambda}_k\|^2 \nonumber \\
    + &\sum_{\chi \in \{t,r\}} \| \boldsymbol{\vartheta}_{\chi} - \boldsymbol{\theta}_{\chi} + \rho \boldsymbol{\mu}_i\|^2.
\end{align}
With some algebraic manipulations, the above problem can be simplified as 
\begin{align}
  \min_{\boldsymbol{\theta}_t, \boldsymbol{\theta}_r} \quad &\sum_{\chi \in \{t,r\}} \Big( \sum_{k \in \mathcal{K}_{\chi}} \|\mathbf{\Phi}_k \boldsymbol{\theta}_{\chi} - \mathbf{u}_k \|^2 + \| \boldsymbol{\theta}_{\chi} - \tilde{\mathbf{u}}_i \|^2 \Big),
\end{align}
where $\mathbf{\Phi}_k = (\mathbf{F} \mathbf{H}_k^n)^T$, $\mathbf{u}_k = \mathbf{p}_k^T + \rho \boldsymbol{\lambda}_k^T$, and $\tilde{\mathbf{u}}_i = \boldsymbol{\vartheta}_{\chi} + \rho \boldsymbol{\mu}_i$. This problem is an unconstrained convex optimization problem. Thus, the optimal solution can be obtained by checking the first-order optimality condition, which is given by 
\begin{align} \label{eqn:narrow_STAR_coupled_solution}
  \boldsymbol{\theta}_{\chi}^\star = \Big( \sum_{k \in \mathcal{K}_{\chi}} \mathbf{\Phi}_k^H \mathbf{\Phi}_k + \mathbf{I}_M \Big)^{-1} \Big( \sum_{k \in \mathcal{K}_{\chi}} \mathbf{\Phi}_k^H \mathbf{u}_k + \tilde{\mathbf{u}}_i \Big).
\end{align}

\subsubsection{Subproblem With Respect to $\{\boldsymbol{\vartheta}_t, \boldsymbol{\vartheta}_r\}$}
The block $\{\boldsymbol{\vartheta}_t, \boldsymbol{\vartheta}_r\}$ is only related to the second term of the objective function and the constraints \eqref{constraint:narrow_coupled_1} and \eqref{constraint:narrow_coupled_2} of the AL problem. Thus, the subproblem with respect to $\{\boldsymbol{\vartheta}_t, \boldsymbol{\vartheta}_r\}$ is given by 
\begin{subequations}
  \begin{align}
      \min_{\boldsymbol{\vartheta}_t, \boldsymbol{\vartheta}_r} \quad &\sum_{\chi \in \{t,r\}} \| \boldsymbol{\vartheta}_{\chi} - \boldsymbol{\theta}_{\chi} + \rho \boldsymbol{\mu}_i \|^2 \\
      \mathrm{s.t.} \quad & \eqref{constraint:narrow_coupled_1}, \eqref{constraint:narrow_coupled_2}.
  \end{align}
\end{subequations}
Although this problem is non-convex due to the non-convex quadratic equality constraint and coupled phase-shift constraint, it has been shown in \cite{wang2022coupled} that it can be solved by iteratively updating the amplitude vector $\tilde{\boldsymbol{\beta}}_i \triangleq [\tilde{\beta}_{i,1}, \dots, \tilde{\beta}_{i,M}]^T, \forall \chi \in \{t,r\}, $ and the phase-shift vector $\tilde{\boldsymbol{\phi}}_i = [e^{j \tilde{\phi}_{i,1}},\dots, e^{j \tilde{\phi}_{i,M}}]^T, \forall \chi \in \{t,r\}$. The closed-form optimal solutions for the phase-shift vector and the amplitude vector when the other is fixed have been given in \cite[Proposition 1]{wang2022coupled} and \cite[Proposition 2]{wang2022coupled}, respectively, which is thus omitted here.

\subsubsection{Update Dual Variables and Penalty Factor}
Problem \eqref{problem:narrow_coupled} can also be solved by exploiting the PDD framework. The corresponding algorithm has a similar structure as \textbf{Algorithm \ref{alg:narrow_PDD}}. The differences with \textbf{Algorithm \ref{alg:narrow_PDD}} in the outer loop for updating dual variables and penalty factor are summarized as follows:
\begin{itemize}
  \item The new constraint violation function is defined as 
  \begin{equation}
    \tilde{h} = \max \left\{  \begin{array}{c}
      \|\mathbf{F} - \mathbf{F}_{\mathrm{RF}} \mathbf{F}_{\mathrm{BB}} \|_{\infty}, \\
      \max_{k,\chi} \| \mathbf{p}_k - \boldsymbol{\theta}_{\chi}^T \mathbf{H}_k^n \mathbf{F} \|_{\infty}, \\
      \max_{\chi} \|\boldsymbol{\vartheta}_{\chi} - \boldsymbol{\theta}_{\chi} \|_\infty
    \end{array}  \right\}.
  \end{equation}
  \item The additional dual variables $\{\boldsymbol{\mu}_t, \boldsymbol{\mu}_r\}$ are updated by 
  \begin{equation}
    \boldsymbol{\mu}_i \leftarrow \boldsymbol{\mu}_i + \frac{1}{\rho} (\boldsymbol{\vartheta}_{\chi} - \boldsymbol{\theta}_{\chi}), \forall \chi \in \{t,r\}.
  \end{equation}
\end{itemize}

\subsubsection{Initialization, Convergence, and Complexity}
To guarantee the performance of the new PDD-based algorithm, its initialization point can be selected as the output of \textbf{Algorithm \ref{alg:narrow_PDD}}, where the variable blocks in addition to $\{\boldsymbol{\vartheta}_t, \boldsymbol{\vartheta}_r\}$ has been well optimized without the coupled phase-shift constraints. The convergence to a stationary point is also guaranteed by the new PDD-based algorithm due to the non-decreasing and upper-bounded objective value over iterations. Compared with \textbf{Algorithm \ref{alg:narrow_PDD}}, the complexity of the new PDD-based algorithm mainly has differences in updating $\{\boldsymbol{\theta}_t, \boldsymbol{\theta}_r\}$ and $\{\boldsymbol{\vartheta}_t, \boldsymbol{\vartheta}_r\}$. Specifically, the complexity of updating according to \eqref{eqn:narrow_STAR_coupled_solution} is $\mathcal{O}(M^3)$. The complexity of updating each entry of $\{\boldsymbol{\vartheta}_t, \boldsymbol{\vartheta}_r\}$ based on the closed-form solution is $\mathcal{O}(6M)$ \cite{wang2022coupled}.

\section{Wideband System} \label{sec:wide}
In this section, we study the wideband STARS-aided THz communication system, where the OFDM technique is adopted to effectively utilize the wideband resources. Specifically, we focus on addressing the issue of beam split caused by the mismatch between the \emph{frequency-dependent} spatial wideband effect and the \emph{frequency-independent} beamforming structures at the BS and STARS. 

\subsection{Wideband Beam Split}
Fig. \ref{fig:wideband_hybrid} illustrates the conventional hybrid beamforming structure for wideband OFDM systems. In this structure, although different digital beamformers can be generated for different subcarriers with different frequencies, the analog beamformer is \emph{frequency-independent} due to the hardware limitation of the PSs. In other words, all subcarriers share the same analog beamformer in the conventional hybrid beamforming structure. However, the wideband THz channel, as detailed in Section \ref{sec:channel_model}, can be significantly \emph{frequency-dependent} due to the \emph{frequency-dependent} array response vectors. Consequently, the conventional hybrid beamforming may result in beam mismatch at different subcarriers, which is referred to as the \emph{beam split} effect. Specifically, when an analog beamformer is designed to steer a beam towards a specific direction at a subcarrier, the beams generated by this analog beamformer at other subcarriers will steer towards other directions, as shown in Fig. \ref{fig:beam_split_conventional}, leading to significant performance degradation. Furthermore, STARS shares a similar property as the analog beamforming structure, which has \emph{frequency-independent} passive beamforming. Therefore, \emph{beam split} effect also exists at STARS.

\begin{figure}[t!]
    \centering
    \subfigure[Conventional hybrid beamforming.]{
        \includegraphics[width=0.4\textwidth]{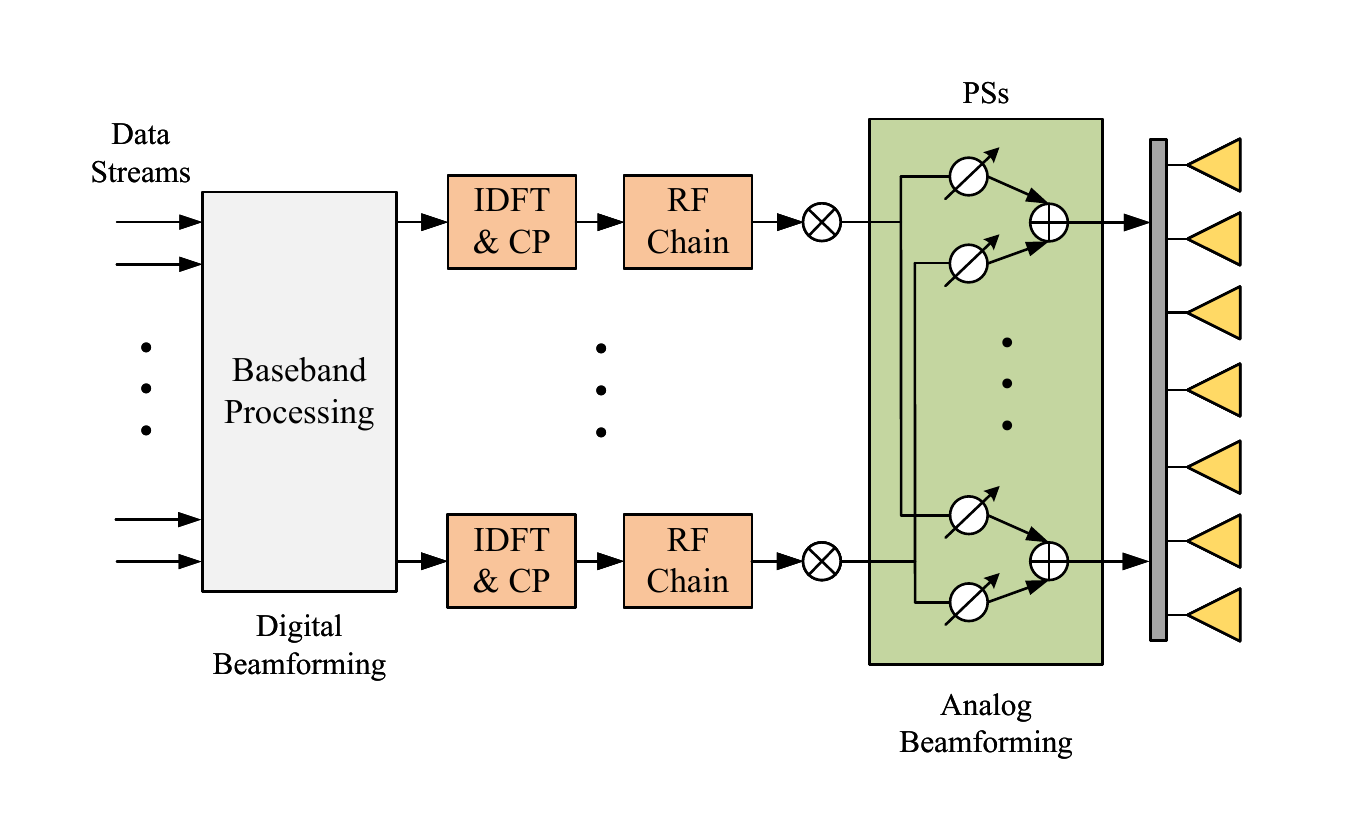}
        \label{fig:wideband_hybrid}
    }
    \subfigure[TTD-based hybrid beamforming.]{
        \includegraphics[width=0.4\textwidth]{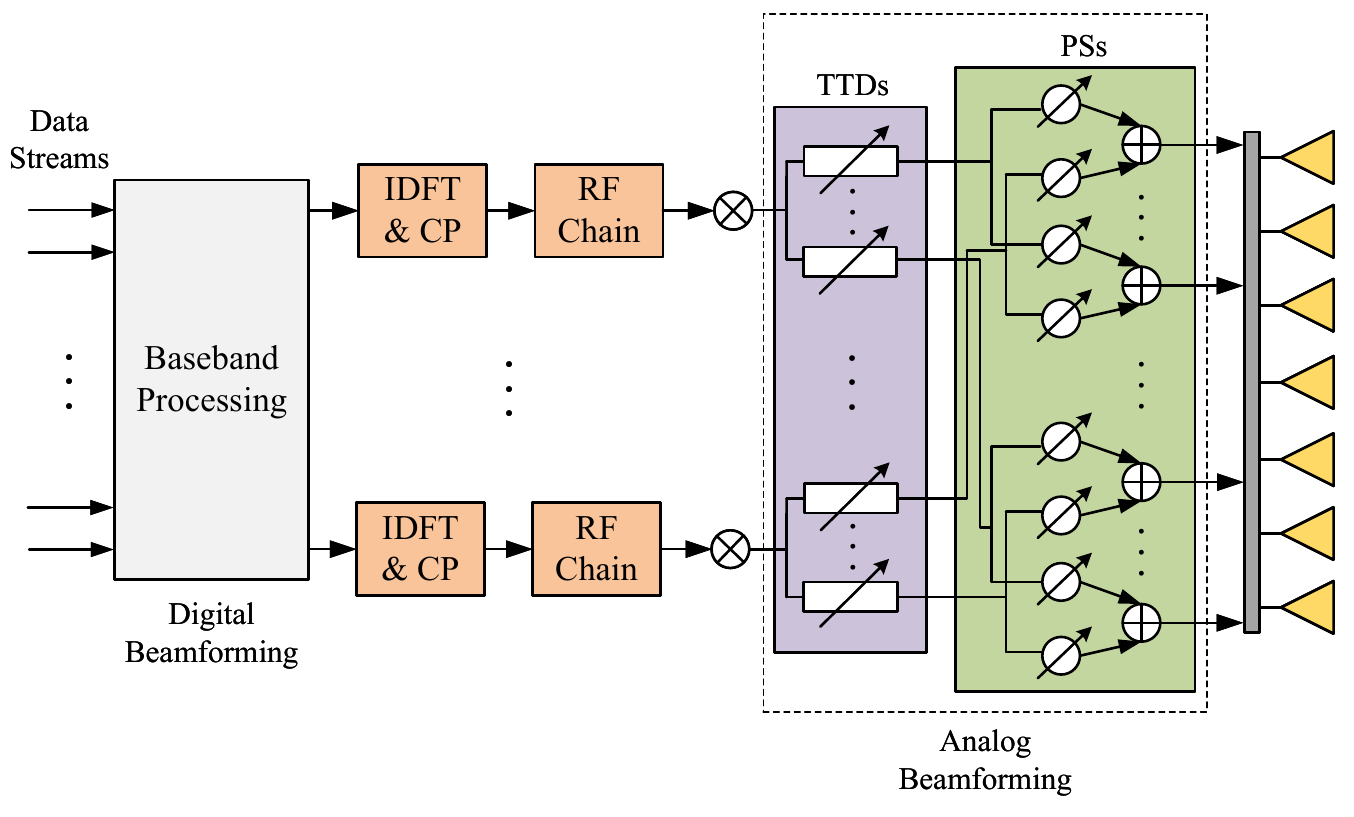}
        \label{fig:wideband_TTD_hybrid}
    }
    \caption{Beamforming structures at the BS in wideband THz communication systems.}
\end{figure}

\subsection{TTD-based Hybrid Beamforming}
As discussed above, the beam split effect mainly arises from the mismatch between the \emph{frequency-dependent} wideband channel and the \emph{frequency-independent} analog beamforming at the BS and passive beamforming at the STARS. To address this issue, a \emph{frequency-dependent} component is required in the system. Given the hardware limitations of the STARS, it is more practical to introduce such a component at the BS. Recently, TTD-based hybrid beamforming structures have been proposed \cite{gao2021wideband, dai2022delay}, as shown in Fig. \ref{fig:wideband_TTD_hybrid}. In this structure, a time-delay network realized by TTDs is introduced between the RF chains and the \emph{frequency-independent} PSs. Unlink PSs, TTDs are capable of achieving \emph{frequency-dependent} phase shifts. For example, a time delay $t$ realized by TTDs becomes a phase shift $e^{-2 \pi f_m t}$ at the subcarrier $m$, thus facilitating the \emph{frequency-dependent} analog beamforming. Consequently, as illustrated in Fig. \ref{fig:beam_split_TD}, the \emph{beam split} effect can be considerably minimized by appropriately designing the time delay of TTDs. However, in wideband STARS-aided THz communication systems, the TTDs need to be configured to mitigate the \emph{beam split} effect caused by both BS and STARS, which requires joint optimization. 

\begin{figure}[t]
    \centering
    \subfigure[]{
        \includegraphics[width=0.4\textwidth]{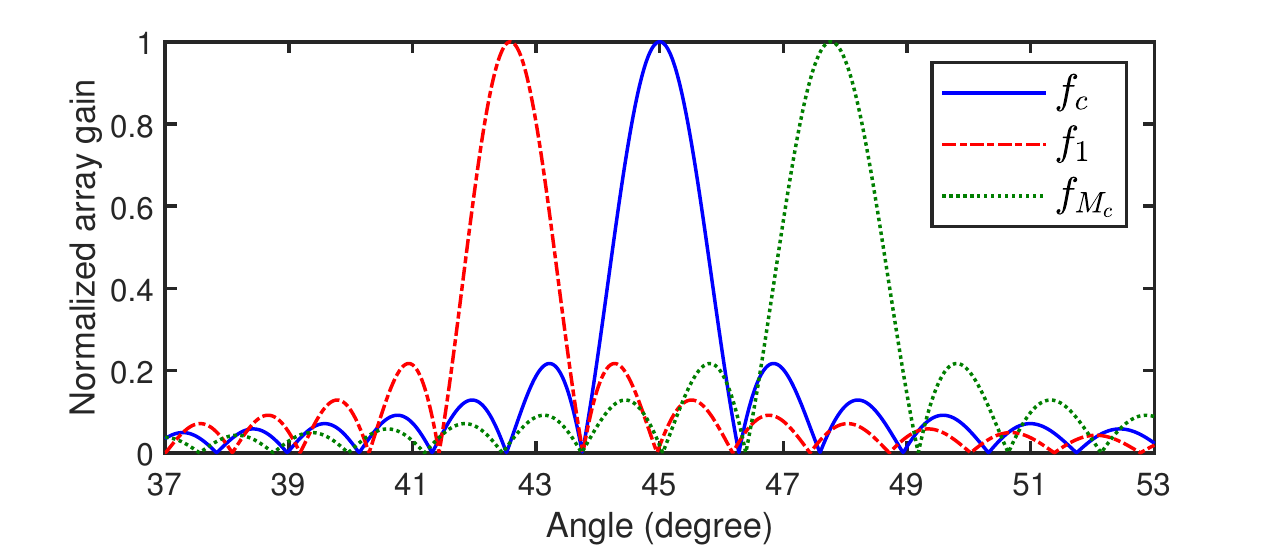}
        \label{fig:beam_split_conventional}
    }
    \subfigure[]{
        \includegraphics[width=0.4\textwidth]{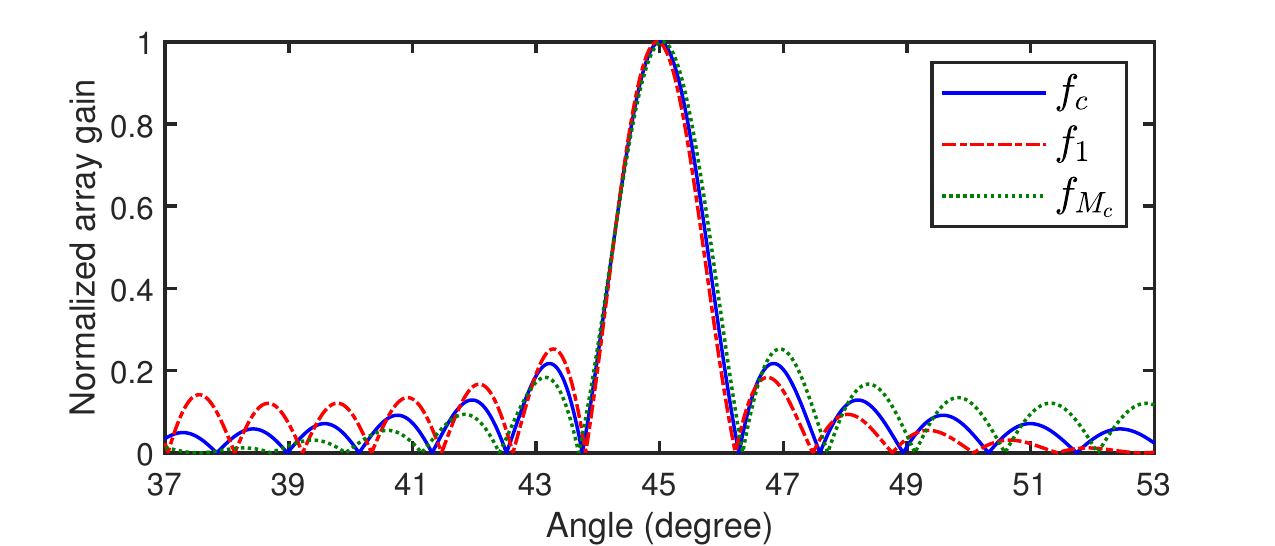}
        \label{fig:beam_split_TD}
    }
    \caption{Normalized array gain achieved by (a) the conventional hybrid beamforming and (b) the TTD-based hybrid beamforming at different frequencies \cite{dai2022delay}, where the desired physical direction is $45^\circ$. Other simulation setup is $N = 128$, $f_c = 0.1$ THz, $W = 10$ GHz, and $M_c = 10$.}
    \label{fig:beam_split}
\end{figure}

In the TTD-based hybrid structure, each RF chain is connected to $N$ PSs via $N_{\mathrm{T}}$ TTDs and each TTD is connected to $\frac{N}{N_{\mathrm{T}}}$ PSs. Thus, there are $N_{\mathrm{RF}} N_{\mathrm{T}}$ TTDs and $N_{\mathrm{RF}} N$ PSs in total.
Let $\mathbf{F}_{\mathrm{PS}} \in \mathbb{C}^{N \times N_{\mathrm{T}} N_{\mathrm{RF}}}$ denote the \emph{frequency-independent} analog beamformer achieved by PSs, $\mathbf{T}_m \in \mathbb{C}^{N_{\mathrm{T}} N_{\mathrm{RF}} \times N_{\mathrm{RF}}}$ denote the \emph{frequency-dependent} analog beamformer achieved by TTDs, $\mathbf{F}_m^{\mathrm{BB}} = [\mathbf{f}_{m,1}^{\mathrm{BB}}, \dots, \mathbf{f}_{m,K}^{\mathrm{BB}}] \in \mathbb{C}^{N_{\mathrm{RF}} \times K }$ denote the digital beamformer for $K$ user at subcarrier $m$, and $\tilde{\mathbf{s}}_m = [s_{m,1},\dots,\tilde{s}_{m,k}]^T \in \mathbb{C}^{K \times 1}$ denote the information symbols for $K$ users at subcarrier $m$. More particularly, the \emph{frequency-dependent} analog beamformer can be expressed as follows:
\begin{equation}
    \mathbf{F}_{\mathrm{PS}} = [\mathbf{F}_1^{\mathrm{PS}},\dots,\mathbf{F}_{N_{\mathrm{RF}}}^{\mathrm{PS}} ],
\end{equation}
where $\mathbf{F}_n^{\mathrm{PS}} = \mathrm{blkdiag}([\mathbf{f}_{n,1}^{\mathrm{PS}},\dots,\mathbf{f}_{n,N_{\mathrm{T}}}^{\mathrm{PS}}]) \in \mathbb{C}^{N \times N_{\mathrm{T}}}$ denote the analog beamformer connected to the $n$-th RF chain via TTDs and $\mathbf{f}_{n,i}^{\mathrm{PS}} \in \mathbb{C}^{\frac{N}{N_{\mathrm{T}}} \times 1}$ denote the corresponding analog beamformer connected to the $i$-th TTD. 
Due to the hardware limitation of PSs, each entry of $\mathbf{f}_{n,i}^{\mathrm{PS}}$ needs to satisfy the following constant-modulus constraint:
\begin{equation}
    |[\mathbf{f}_{n,i}^{\mathrm{PS}} ]_j| = 1, \forall n, i, j.
\end{equation}
The \emph{frequency-independent} analog beamformer can be expressed as follows:
\begin{equation} \label{eqn:wideband_Tm}
    \mathbf{T}_m = \mathrm{blkdiag} \left( [ e^{-j 2 \pi f_m \mathbf{t}_1},\dots, e^{-j 2 \pi f_m \mathbf{t}_{N_{\mathrm{RF}}}}] \right),
\end{equation} 
where $\mathbf{t}_n = [t_{n,1},...,t_{n, N_{\mathrm{T}}}]^T \in \mathbb{C}^{N_{\mathrm{T}} \times 1}$ denotes the time delays realized by TTDs connected to the $n$-th RF chain, where $t_{n,i} \ge 0$. One can observe from \eqref{eqn:wideband_Tm} that the phase shifts of matrices $\mathbf{T}_m$ and $\mathbf{T}_n, \forall n \neq m,$ are coupled with each other since they are realized by a common set of time delays. 
The transmit signal at the subcarrier $m$ through the TTD-based hybrid beamforming is expressed as follows:
\begin{equation}
    \tilde{\mathbf{x}}_m = \mathbf{F}_{\mathrm{PS}} \mathbf{T}_m \mathbf{F}_m^{\mathrm{BB}} \tilde{\mathbf{s}}_m = \mathbf{F}_{\mathrm{PS}} \mathbf{T}_m \sum_{k \in \mathcal{K}} \mathbf{f}_{m,k}^{\mathrm{BB}} \tilde{s}_{m,k}.
\end{equation} 
Assuming that $\tilde{\mathbf{s}}_m$ is an independent complex Gaussian signal, the covariance matrix of $\tilde{\mathbf{x}}_m$ can be obtained as $\tilde{\mathbf{Q}}_m = \mathbb{E}[ \tilde{\mathbf{x}}_m \tilde{\mathbf{x}}_m^H ] = \mathbf{F}_{\mathrm{PS}} \mathbf{T}_m \mathbf{F}_m^{\mathrm{BB}} (\mathbf{F}_{\mathrm{PS}} \mathbf{T}_m \mathbf{F}_m^{\mathrm{BB}})^H$. We assume the power constraint for each subcarrier to be the same, resulting in the following power constraint:
\begin{equation}
    \mathrm{tr}(\tilde{\mathbf{Q}}_m) =  \|\mathbf{F}_{\mathrm{PS}} \mathbf{T}_m \mathbf{F}_m^{\mathrm{BB}}\|_F^2 \le P_t.
\end{equation}
The received signal at subcarrier $m$ for user $k, \forall k \in \mathcal{K}_{\chi}, \chi \in \{t,r\},$ is given by 
\begin{align} \label{eqn:wideband_receive_signal}
    \tilde{y}_{m,k} = &\mathbf{h}_{m,k}^w \tilde{\mathbf{x}}_m + \tilde{n}_{m,k}  \nonumber \\
    = & \mathbf{v}_{m,k}^w \mathbf{\Theta}_{\chi} \mathbf{G}_m^w \mathbf{F}_{\mathrm{PS}} \mathbf{T}_m \sum_{k \in \mathcal{K}} \mathbf{f}_{m,k}^{\mathrm{BB}} \tilde{s}_{m,k} + \tilde{n}_{m,k}, \nonumber \\
    = & \underbrace{\boldsymbol{\theta}_{\chi}^T \mathbf{H}_{m,k}^w \mathbf{F}_{\mathrm{PS}} \mathbf{T}_m \mathbf{f}_{m,k}^{\mathrm{BB}} \tilde{s}_{m,k}}_{\text{desired signal}} \nonumber \\
    & \quad + \underbrace{\sum_{i \in \mathcal{K}, i \neq k} \boldsymbol{\theta}_{\chi}^T \mathbf{H}_{m,k}^w \mathbf{F}_{\mathrm{PS}} \mathbf{T}_m \mathbf{f}_{m,i}^{\mathrm{BB}} \tilde{s}_{i,k}}_{\text{inter-user interference}} + \tilde{n}_{m,k},
\end{align}  
where $\mathbf{H}_{m,k}^w = \mathrm{diag}(\mathbf{v}_{m,k}^w) \mathbf{G}_m^w$ denotes the cascaded channel from BS to user $k$ at subcarrier $m$, and $\tilde{n}_{m,k} \sim \mathcal{CN}(0,\sigma_{m,k}^2)$ denotes the additive complex Gaussian noise.

\subsection{Problem Formulation}
Similarly, we aim to maximize the SE and EE of wideband STARS-aided THz communication systems. According to \eqref{eqn:wideband_receive_signal}, the achievable rate for user $k, \forall k \in \mathcal{K}_{\chi}, \chi \in \{t,r\},$ at subcarrier $m$ is given by 
\begin{align}
    &\tilde{R}_{m,k} \nonumber \\[-0.1em]
    &= \log_2 \left( 1 + \frac{| \boldsymbol{\theta}_{\chi}^T \mathbf{H}_{m,k}^w \mathbf{F}_{\mathrm{PS}} \mathbf{T}_m \mathbf{f}_{m,k}^{\mathrm{BB}} |^2}{\sum_{i \in \mathcal{K}, i \neq k} |\boldsymbol{\theta}_{\chi}^T \mathbf{H}_{m,k}^w \mathbf{F}_{\mathrm{PS}} \mathbf{T}_m \mathbf{f}_{m,i}^{\mathrm{BB}}|^2 + \sigma_{m,k}^2 } \right).
\end{align}  
Thus, the SE of the wideband STARS-aided THz-OFDM system is given by 
\begin{equation}
    \tilde{f}_{\mathrm{SE}} = \mu \sum_{m \in \mathcal{M}_c} \sum_{k \in \mathcal{K}} \tilde{R}_{m,k},
\end{equation}
where $\mathcal{M}_c = \{1,\dots,M_c\}$, $\mu = 1/(M_c+ L_{\mathrm{CP}})$, and $L_{\mathrm{CP}} \ge \max_{k \in \mathcal{K}} \{ Q_k \}$ denotes the length of the cyclic prefix (CP) of the OFDM system. Then, the rate-dependent power consumption can be modeled as follows:
\begin{equation}
    \tilde{P} = \frac{1}{M_c} \sum_{m \in \mathcal{M}_c} \|\mathbf{F}_{\mathrm{PS}} \mathbf{T}_m \mathbf{F}_m^{\mathrm{BB}} \|_F^2 + \xi \tilde{f}_{\mathrm{SE}} + \tilde{P}_c,
\end{equation} 
where $\frac{1}{M_c} \sum_{m \in \mathcal{M}_c} \|\mathbf{F}_{\mathrm{PS}} \mathbf{T}_m \mathbf{F}_m^{\mathrm{BB}} \|_F^2$ is the average power consumption over all subcarriers, and $\tilde{P}_c$ denotes the rate-independent power consumption given as follows:
\begin{align}
    \tilde{P}_c = P_{\mathrm{BS}} &+ P_{\mathrm{BB}} + N_{\mathrm{RF}} P_{\mathrm{RF}} + N_{\mathrm{RF}} N_{\mathrm{T}} P_{\mathrm{TTD}} \nonumber \\
    &+  N_{\mathrm{RF}} N P_{\mathrm{RS}} + P_{\mathrm{STAR}} + K P_{\mathrm{UE}}.
\end{align} 
Here, $P_{\mathrm{TTD}}$ denotes the power consumption of each TTD. Then, the EE of the wideband system can be expressed as 
\begin{equation}
    \tilde{f}_{\mathrm{EE}} = \frac{\tilde{f}_{\mathrm{SE}}}{\tilde{P}} = \frac{\tilde{f}_{\mathrm{SE}}}{\frac{1}{M_c} \sum_{m \in \mathcal{M}_c} \|\mathbf{F}_{\mathrm{PS}} \mathbf{T}_m \mathbf{F}_m^{\mathrm{BB}} \|_F^2 + \xi \tilde{f}_{\mathrm{SE}} + \tilde{P}_c }
\end{equation}

The general optimization problem for SE and EE maximization can be formulated as follows:
\begin{subequations} \label{problem:wideband_tradeoff}
    \begin{align}
        \max_{\scriptstyle \mathbf{F}_m^\mathrm{BB}, \mathbf{f}_{n,i}^{\mathrm{PS}}, t_{n,i} \atop \scriptstyle \boldsymbol{\theta}_t, \boldsymbol{\theta}_r } & \frac{\tilde{f}_{\mathrm{SE}}}{ w \big(\frac{1}{M_c} \sum_{m \in \mathcal{M}_c} \|\mathbf{F}_{\mathrm{PS}} \mathbf{T}_m \mathbf{F}_m^{\mathrm{BB}} \|_F^2 + \xi \tilde{f}_{\mathrm{SE}} \big) + \tilde{P}_c } \\
        \label{constraint:wideband_power}
        \mathrm{s.t.} \quad & \|\mathbf{F}_{\mathrm{PS}} \mathbf{T}_m \mathbf{F}_m^{\mathrm{BB}} \|_F^2 \le P_t, \forall m, \\
        \label{constraint:wideband_0}
        & \boldsymbol{\theta}_{\chi} \in \mathcal{F}, \forall \chi \in \{t,r\}, \\
        \label{constraint:wideband_1}
        & |[\mathbf{f}_{n,i}^{\mathrm{PS}} ]_j| = 1, \forall n, i, j, \\
        \label{constraint:wideband_2}
        & t_{n,i} \ge 0, \forall n, i.
    \end{align}
\end{subequations}
Compared to the problem \eqref{problem:narrow_tradeoff_2} encountered in narrowband systems, the problem \eqref{problem:wideband_tradeoff} poses a greater challenge. This is primarily because it involves additional \emph{frequency-dependent} analog beamformers $\mathbf{T}_m$ realized by a common set of time delays $t_{n,i}$.

\subsection{Proposed Solution}
In this subsection, we propose a new PDD-based algorithm for solving problem \eqref{problem:wideband_tradeoff}, with a particular focus on optimizing the time delays $t_{n,i}$ of each TTD. In a manner similar to narrowband systems, we define auxiliary variables $\tilde{\eta}$, $\tilde{a}$, $\tilde{b}$, $\tilde{r}_{m,k}$, $\tilde{\mathbf{F}}_m$, and $\tilde{\mathbf{p}}_{m,k}$. Then, problem \eqref{problem:wideband_tradeoff} be transferred into the following equivalent form:      
\begin{subequations} \label{problem:wideband_tradeoff_2}
    \begin{align}
        & \hspace{-0.6cm} \max_{
            \begin{subarray}{c}
                \mathbf{F}_m^{\mathrm{BB}}, \mathbf{f}_{n,i}^{\mathrm{PS}}, t_{n,i}, \tilde{\mathbf{F}}_m, \tilde{\mathbf{p}}_{m,k}\\
                \boldsymbol{\theta}_t, \boldsymbol{\theta}_r, \tilde{\eta}, \tilde{a}, \tilde{b}, \tilde{r}_{m,k}\\
            \end{subarray}
            } \quad \tilde{\eta} \\
        \label{constraint:wideband_transformed_1}
        \hspace{0.4cm} \mathrm{s.t.} \quad & \tilde{\eta} \le \frac{\tilde{a}^2}{\tilde{b}}, \\
        \label{constraint:wideband_transformed_2}
        & \tilde{a}^2 \le \mu \sum_{m,k} \tilde{r}_{m,k}, \\
        \label{constraint:wideband_transformed_3}
        & w \big( \frac{1}{M_c} \sum_m \|\tilde{\mathbf{F}}_m\|_F^2 + \xi \mu \sum_{m,k} \tilde{r}_{m,k} \big) + \tilde{P_c} \le \tilde{b}, \\
        \label{constraint:wideband_transformed_4}
        & \tilde{r}_{m,k} \le \tilde{R}_{m,k}(\tilde{\mathbf{p}}_{m,k}), \forall m, k, \\
        \label{constraint:wideband_transformed_5}
        & \| \tilde{\mathbf{F}}_m \|_F^2 \le P_t, \forall m, \\
        \label{constraint:wideband_equality_1}
        & \tilde{\mathbf{F}}_m = \mathbf{F}_{\mathrm{PS}} \mathbf{T}_m \mathbf{F}_m^{\mathrm{BB}}, \forall m, \\
        \label{constraint:wideband_equality_2}
        & \tilde{\mathbf{p}}_{m, k} = \boldsymbol{\theta}_{\chi}^T \mathbf{H}_{m,k}^w \tilde{\mathbf{F}}_m, \forall m, k, \chi, \\
        & \eqref{constraint:wideband_0}-\eqref{constraint:wideband_2}.
    \end{align}
\end{subequations}
Here, $\tilde{R}_{m,k}(\tilde{\mathbf{p}}_{m,k})$ is defined as 
\begin{equation}
    \tilde{R}_{m,k}(\tilde{\mathbf{p}}_{m,k}) = \log_2 \left( 1 + \frac{|\tilde{p}_{k,k}^m|^2 }{ \sum_{i \in \mathcal{K}, i \neq k} |\tilde{p}_{k,i}^m|^2 + \sigma_{m,k}^2 } \right),
\end{equation}
where $\tilde{p}_{k,i}^m$ denotes the $i$-th entry of $\tilde{\mathbf{p}}_{m,k}$.   
Then, by introducing the dual variables $\tilde{\mathbf{\Psi}}_m$ and $\tilde{\boldsymbol{\lambda}}_{m,k}, \forall m \in \mathcal{M}_c, k \in \mathcal{K},$ for the equality constraints \eqref{constraint:wideband_equality_1} and \eqref{constraint:wideband_equality_2}, respectively, the following AL problem of \eqref{problem:wideband_tradeoff_2} can be formulated:
\begin{subequations}
    \begin{align}
        & \hspace{-1.8cm} \max_{
            \begin{subarray}{c}
            \mathbf{F}_m^{\mathrm{BB}}, \mathbf{f}_{n,i}^{\mathrm{PS}}, t_{n,i}, \tilde{\mathbf{F}}_m, \tilde{\mathbf{p}}_{m,k}, \\
            \boldsymbol{\theta}_t, \boldsymbol{\theta}_r, \tilde{\eta}, \tilde{a}, \tilde{b}, \tilde{r}_{m,k}\\
        \end{subarray}
        } \quad \tilde{\eta} - P_\rho\\ 
        \mathrm{s.t.} \quad & \eqref{constraint:wideband_0}-\eqref{constraint:wideband_2}, \eqref{constraint:wideband_transformed_1} - \eqref{constraint:wideband_transformed_5}.
    \end{align}
\end{subequations}
where
\begin{align}
    P_\rho = &\sum_{m} \frac{1}{2\rho} \|\tilde{\mathbf{F}}_m - \mathbf{F}_{\mathrm{PS}} \mathbf{T}_m \mathbf{F}_m^{\mathrm{BB}} + \rho \tilde{\mathbf{\Psi}}_m \|_F^2 \nonumber \\
        &- \sum_{m,k,\chi} \frac{1}{2\rho} \|\tilde{\mathbf{p}}_{m,k} - \boldsymbol{\theta}_{\chi}^T \mathbf{H}_{m,k}^w \tilde{\mathbf{F}}_m + \rho \tilde{\boldsymbol{\lambda}}_{m,k} \|^2.
\end{align} 
To address the AL problem presented above, the BCD is employed, whereby the optimization variable is divided into five blocks: $\{\tilde{\mathbf{F}}_m, \tilde{\mathbf{p}}_{m,k}, \tilde{\eta}, \tilde{a}, \tilde{b}, \tilde{r}_{m,k}\}$, $\{\boldsymbol{\theta}_t, \boldsymbol{\theta}_r\}$, $\{\mathbf{f}_{n,i}^{\mathrm{PS}}\}$, $\{\mathbf{F}_m^{\mathrm{BB}}\}$, and $\{t_{n,i}\}$. Notably, the subproblems associated with the first four blocks have the same structure as those encountered in narrowband systems and can be solved via the methods outlined in Sections \ref{sec:narrow_ideal_solution} and \ref{sec:narrow_coupled_solution} for independent and coupled phase-shift STARSs, respectively. Furthermore, the dual variables and penalty factor can also be updated similarly in the outer loop of the PDD framework. Therefore, we focus on solving the subproblem with respect to $\{t_{n,i}\}$ in the following.

\subsubsection{Subproblem With Respect to $\{t_{n,i}\}$}

The block $\{t_{n,i}\}$ only appears in the penalty term $P_\rho$ and the constraint \eqref{constraint:wideband_2}. Thus, the corresponding subproblem is given by
\begin{subequations} \label{problem:wideband_Tm}
    \begin{align}
        \min_{t_{n,i}} \quad & \sum_{m \in \mathcal{M}_c}  \|\tilde{\mathbf{F}}_m - \mathbf{F}_{\mathrm{PS}} \mathbf{T}_m \mathbf{F}_m^{\mathrm{BB}} + \rho \tilde{\mathbf{\Psi}}_m \|_F^2 \\
        \mathrm{s.t.} \quad & t_{n,i} \ge 0, \forall n, i.
    \end{align}
\end{subequations}
The problem described above can be converted into an unconstrained optimization problem. Specifically, the objective function can be reformulated as a function of $t_{n,i}$ by substituting equation \eqref{eqn:wideband_Tm}. From equation \eqref{eqn:wideband_Tm}, it is evident that $\mathbf{T}_m$ is a periodic function of $t_{n,i}$, with a period of $1/f_m$. Typically, the ratio of these periods is rational, indicating that the objective function is also periodic with respect to $t_{n,i}$. Consequently, any negative value of $t_{n,i}$ can be replaced by a positive value $t_{n,i}' \geq 0$ that produces the same objective value. Hence, the constraint $t_{n,i} \geq 0$ can be removed without affecting the solution. Now, with no constraint on $t_{n,i}$, problem \eqref{problem:wideband_Tm} can be formulated as an unconstrained optimization problem. Therefore, the optimization variables $t_{n,i}$ can be optimized directly using the quasi-Newton method \cite{nocedal2006numerical}. 

\subsubsection{Initialization, Convergence, and Complexity}

The initial optimization variables of the new PDD-based algorithm for wideband systems are generated as follows. Firstly, to initialize $\mathbf{f}_{n,i}^{\mathrm{PS}}$ and $t_{n,i}$, the $N_{\mathrm{RF}}$ strongest paths are selected from the total $L$ paths between the BS and STARS. Then, they are initialized following the design principle proposed in \cite{dai2022delay}. More specifically, given the $n$-th strongest path with direction $\varphi_n$, $\mathbf{f}_{n,i}^{\mathrm{PS}}$ is initialized as \cite[Eq. (30)]{dai2022delay}  
\begin{align}
    \mathbf{f}_{n,i}^{\mathrm{PS},\text{init}}  = e^{j \pi (i-1) \delta \sin \varphi_n } [\mathbf{b}(f_c, \varphi_n)]_{ (i-1) \delta+1 : i \delta }, 
    \forall i.
\end{align}
where $\delta = N / N_{\mathrm{T}}$. 
Then, the time delay for the $i$-th TTD connected to the $n$-th RF chain is initialized as \cite[Eq. (31)]{dai2022delay}  
\begin{equation}
    t_{n,i}^{\text{init}} = \begin{cases}
        (i-1) \frac{\delta \sin \varphi_n}{2 f_c}, & \sin \varphi_n \ge 0, \\
        (i-1) \frac{\delta \sin \varphi_n}{2 f_c} + (N_{\mathrm{T}} - 1) \left| \frac{\delta \sin \varphi_n}{2 f_c} \right|, & \sin \varphi_n < 0.
    \end{cases}
\end{equation} 
Then, the digital beamformers $\mathbf{F}_{\mathrm{BB}, m}$ and the STARS coefficients $\{\boldsymbol{\theta}_t, \boldsymbol{\theta}_r\}$ are randomly initialized within the feasible set. The auxiliary variables are initialized such that the corresponding equality constraints are satisfied. 

Similarly, the new developed PDD-based algorithm is also guaranteed to converge to a stationary point. The complexity of the new algorithm is analyzed as follows. Firstly, when updating the block $\{t_{n,i}\}$, the number of unconstrained optimization variables $t_{n,i}$ is equivalent to the number of TTDs, which is $N_{\mathrm{RF}} N_{\mathrm{T}}$. By using Broyden-Fletcher-Goldfarb-Shanno (BFGS) formula for updating the approximation of the Hessian matrix in the quasi-Newton method, the complexity of each iteration is $\mathcal{O}(N_{\mathrm{RF}}^2 N_{\mathrm{T}}^2)$ \cite{nocedal2006numerical}. Then, since the other variable blocks can be solved by the methods outlined in Sections \ref{sec:narrow_ideal_solution} and \ref{sec:narrow_coupled_solution}, their corresponding complexity can be analyzed similarly. Therefore, we omit it here.
\section{Numerical Results} \label{sec:result}

In this section, the numerical results obtained through Monte Carlo simulations are provided to evaluate the performance of the proposed STARS-aided THz wireless communication system in both narrowband and wideband systems. Fig. \ref{fig:simulation_setup} illustrates the considered three-dimensional simulation setup. In particular, it is assumed that a BS equipped with $N=128$ antennas and $N_{\mathrm{RF}} = 4$ RF chains is $10$ m away from the STARS. There are $K=4$ communication users located on half-circles centered at the STARS with a radius of $3$ m. The number of paths between the BS and STARS and between the STARS and users is assumed to be $L=4$ and $L_k = 4$, respectively. The azimuth and elevation physical angles of channel paths are randomly generated following $\mathcal{U}[-\frac{\pi}{2}, \frac{\pi}{2}]$. The transmit and receive antenna gain are set to $G_t = 25$ dBi and $G_r = 20$ dBi, respectively. The noise power density is assumed to be $-174$ dBm/Hz. The frequency-dependent medium absorption coefficient $k(f)$ in THz pathloss model is obtained from the high-resolution transmission (HITRAN) database \cite{rothman2013hitran2012}.

For power consumption, the rate-dependent consumption factor $\xi$ is set to $0.1$ W/(bit/s/Hz). For rate-independent power consumption, the practice values are adopted as $P_{\mathrm{BS}} = 3$ W \cite{bjornson2015optimal}, $P_{\mathrm{BB}} = 300$ mW \cite{mendez2016hybrid}, $P_{\mathrm{RF}} = 200$ mW \cite{dai2022delay}, $P_{\mathrm{PS}} = 30$ mW \cite{mendez2016hybrid}, $P_{\mathrm{TTD}} = 100$ mW \cite{cho2018true}, and $P_{\mathrm{UE}} = 100$ mW \cite{bjornson2015optimal}. For STARSs, the power consumption of each PIN diode and control circuit is $P_{\mathrm{PIN}} = 0.33$ mW and $P_{\mathrm{circ}} = 10$ W, respectively \cite{9206044}. Furthermore, it is assumed that the maximum tolerable error of the amplitudes and phase shifts is $0.005$ and $1^\circ$, respectively. Thus, according to the model proposed in Section \ref{sec:STAR_power}, the power consumption of each element of the independent phase-shift STARS and the coupled phase-shift STARS is $3.63$ mW and $2.64$ mW, respectively. 

\begin{figure}[t!]
    \centering
    \includegraphics[width=0.45\textwidth]{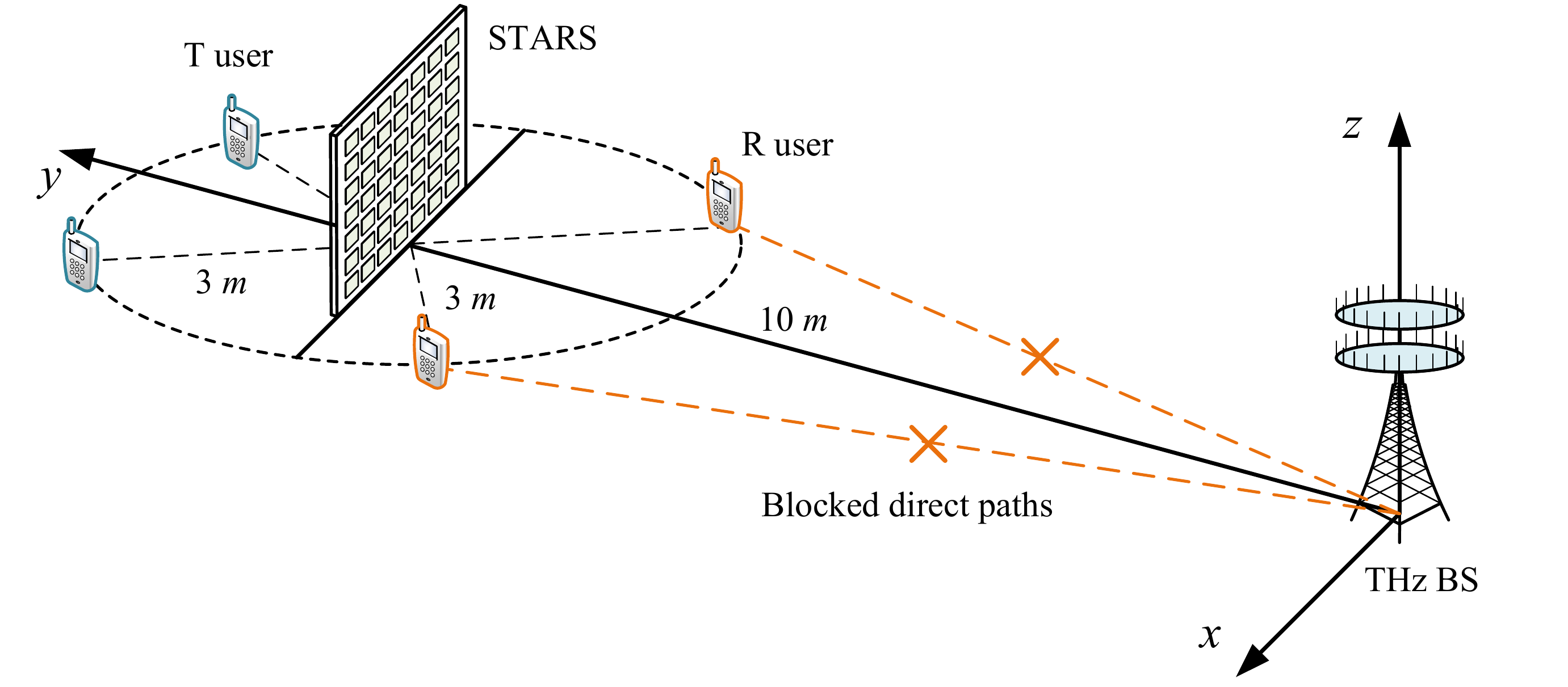}
    \caption{The simulation setup.}
    \label{fig:simulation_setup}
\end{figure}

For the proposed algorithms, the convergence thresholds are set to $10^{-3}$. The initial penalty factor of the PDD-based algorithms is set to $\rho = 10^3$ and its reduction factor is set to $c = 0.6$. The convex problems are solved by the CVX toolbox \cite{cvx}. The golden-section search method for solving \eqref{problem:narrow_theta_2} and the quasi-Newton method for solving \eqref{problem:wideband_Tm} are implemented by the MATLAB function $\mathtt{fminbnd}$ and $\mathtt{fminunc}$, respectively. The following simulation results are obtained by averaging over $100$ random channel realizations. In particular “STARS-i” and “STARS-c” represent the independent and couple phase-shift STARSs, respectively, “HB” represents the conventional hybrid beamforming, and “TTD” represents the TTD-based hybrid beamforming.

\subsection{Narrowband System}
We first investigate the performance of a narrowband STARS-aided THz communication system that operates at a frequency of $0.1$ THz and has a bandwidth of $100$ MHz \cite{su2023wideband}. For performance comparison, we consider the following benchmark schemes:

\begin{itemize}
    \item \textbf{Full-digital (FD) beamforming}: In this scheme, each antenna at the BS is linked to an RF chain in this scheme, necessitating $N$ RF chains. The transmit signal then becomes $\mathbf{x}_{\mathrm{FD}}[n] = \mathbf{F}_{\text{FD}} \mathbf{s}[n]$, where $\mathbf{F}_{\text{FD}} \in \mathbb{C}^{N \times K}$ represents the unconstrained FD beamformer.
    \item \textbf{Conventional RIS}: This scheme employs two $M/2$-element RISs, one for reflection and one for transmission. These RISs are placed next to each other at the same location as the STARS. The power consumption of the conventional RIS can be calculated similarly to the STARS. But the PIN diodes are only required to control the phase shifts. When the maximum tolerable error of phase shifts is $1^\circ$, the power consumption of each element is $1.32$ mW. 
\end{itemize}

\subsubsection{Spectral efficiency versus $P_t$}

In Fig. \ref{fig:narrow_SE_power}, we investigate the achieved SE under different maximum transmit power $P_t$ using different schemes when $w=0$ (SE maximization) and $w=1$ (EE maximization). We set $M = 6 \times 6$. It can be observed for both SE and EE maximization, the obtained SE increases with $P_t$ when $P_t \le 35$ dBm. However, as the transmit power $P_t$ increases beyond $35$ dBm, the SE obtained through EE maximization plateaus, while the SE obtained through SE maximization continues to increase. The reason for this phenomenon is that, at high values of $P_t$, only a fraction of the available power is utilized for maximizing EE, whereas maximizing SE always aims to utilize all available transmit power. Furthermore, although utilizing fewer RF chains, the hybrid beamforming achieves a comparable performance to the fully-digital beamforming. Finally, the STARS approach achieves a substantial performance gain over conventional RIS due to its utilization of all elements for both transmission and reflection. As a result, it can generate more precise directional beams towards users, leading to higher array gain and superior inter-user interference mitigation compared to conventional RIS.

\begin{figure}[t!]
    \centering
    \includegraphics[width=0.4\textwidth]{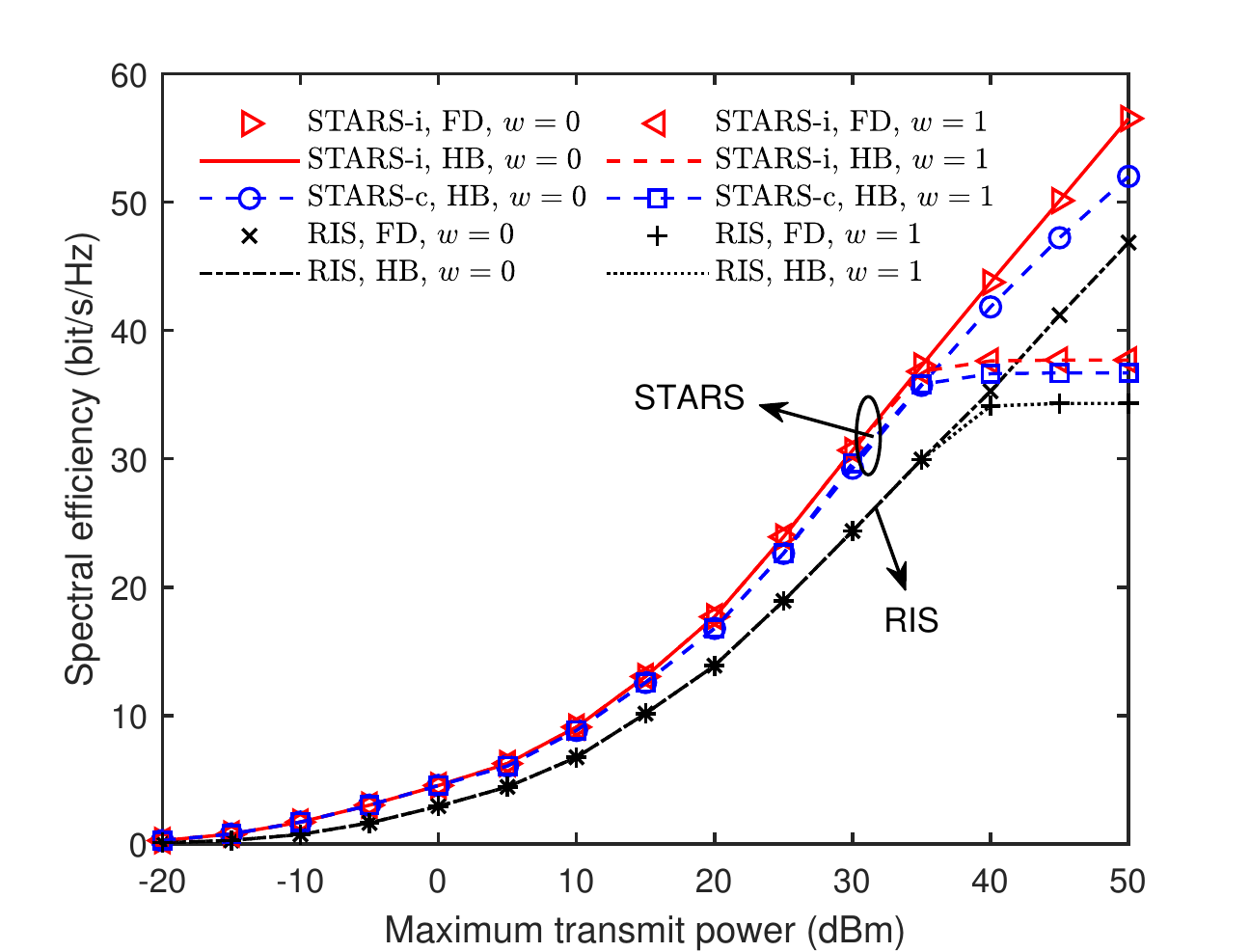}
    \caption{Spectral efficiency versus maximum transmit power $P_t$ for $M = 6 \times 6$ in the narrowband system.}
    \label{fig:narrow_SE_power}
\end{figure}
\begin{figure}[t!]
    \centering
    \includegraphics[width=0.4\textwidth]{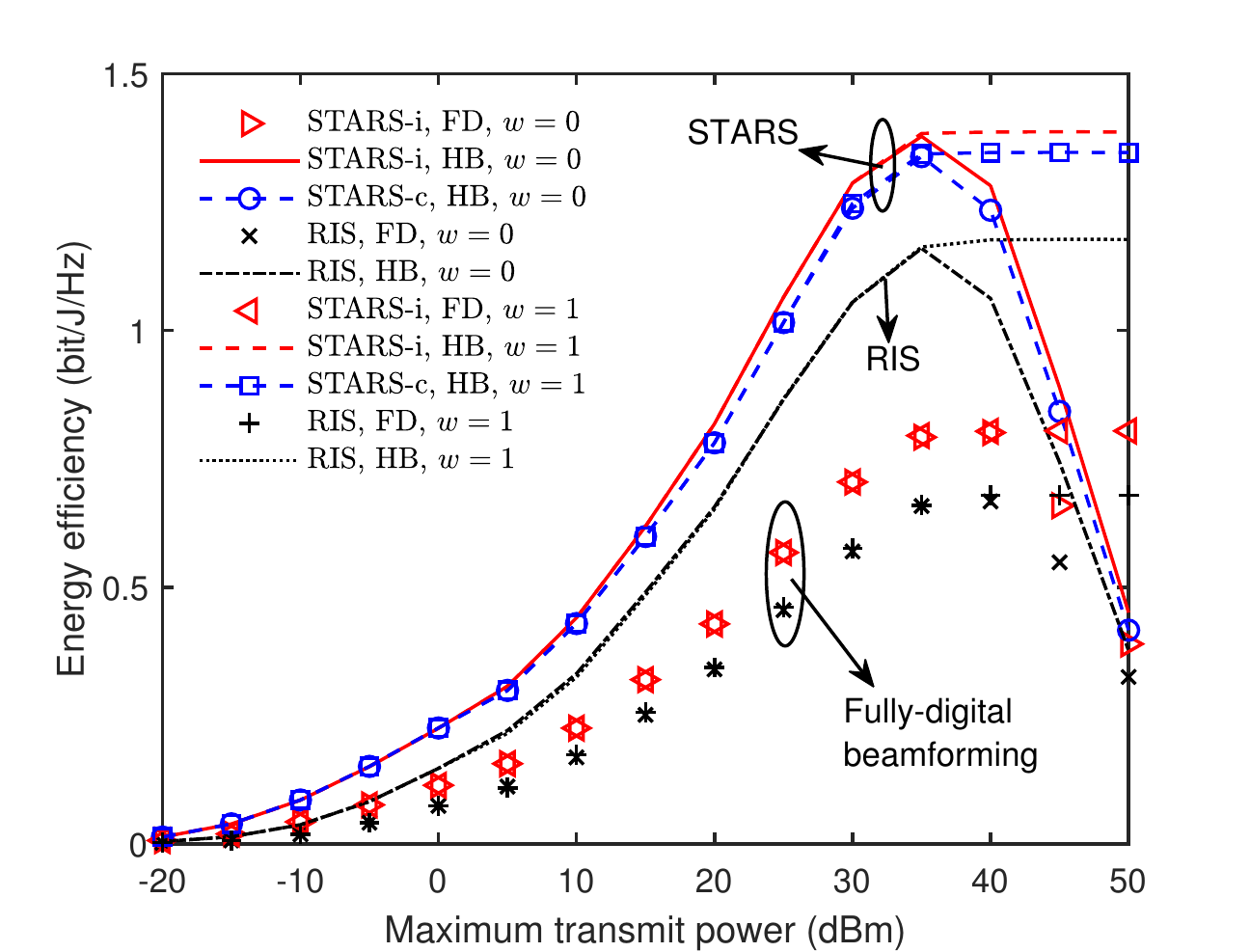}
    \caption{Energy efficiency versus maximum transmit power $P_t$ for $M = 6 \times 6$ in the narrowband system.}
    \label{fig:narrow_EE_power}
\end{figure}

\subsubsection{Energy efficiency versus $P_t$}
In Fig. \ref{fig:narrow_EE_power}, we study the achieved EE under different maximum transmit power $P_t$ using different schemes when $w=0$ (SE maximization) and $w=1$ (EE maximization). We set $M=6 \times 6$. For the results obtained by the EE maximization, it can be observed that the EE first increases with $P_t$ and finally becomes constant. This is because the EE is not a monotonically increasing function of the transmit power $P_t$ but has a finite upper bound. In contrast, SE maximization achieves the same EE performance in the low-power region, but results in significant drops in EE in the high-power region due to the different power utilization strategies of SE and EE maximization. Furthermore, hybrid beamforming consistently outperforms fully-digital beamforming. This is because, in the hybrid beamforming structure, a large number of RF chains with a large power consumption are replaced with a large number of PSs with much lower power consumption. Although these PSs cannot adjust the amplitude of signals as the RF chains, hybrid beamforming can still achieve comparable SE as fully-digital beamforming, but uses much less power. Finally, the superiority of both independent and coupled phase-shift STARS over conventional RIS can be observed.

\begin{figure}[t!]
    \centering
    \includegraphics[width=0.4\textwidth]{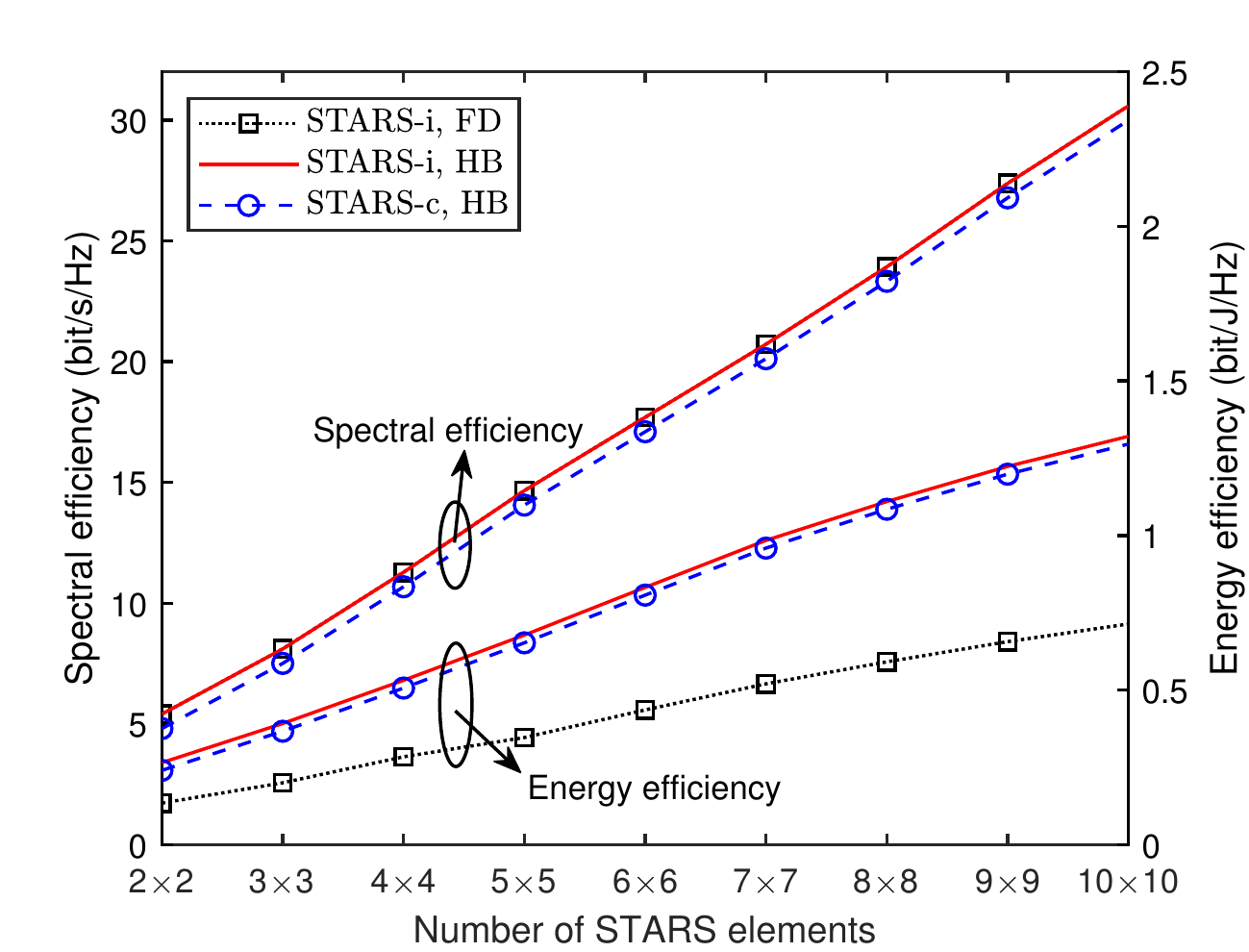}
    \caption{Performance versus the number of STARS elements $M$ for $P_t = 20$ dBm in the narrowband system.}
    \label{fig:narrow_EE_element}
\end{figure}

\subsubsection{Impact of $M$}

In Fig. \ref{fig:narrow_EE_element}, we study the impact of the number of STARS elements on the maximum SE and EE. We set $P_t = 20$ dBm. As can be observed, both SE and EE monotonically increase with the number of RIS elements. This is because the larger number of STARS elements provides more degrees of freedom to achieve the higher SE but the power consumption introduced by the additional elements is relatively low. 

\subsection{Wideband System}
We continue to investigate a wideband OFDM system that operates at a frequency of $f_c = 0.1$ THz and has a bandwidth of $10$ GHz \cite{su2023wideband}. The number of TTDs for each RF chain is set to $N_{\mathrm{T}} = 8$. The number of subcarriers is set to $M_c = 10$. The length of CP is set to $L_{\mathrm{CP}} = 4$. Apart from the FD beamforming and conventional RIS, the following benchmark scheme is also considered:
\begin{itemize}
    \item \textbf{Conventional hybrid beamforming}: In this scheme, the analog beamforming is achieved by exploiting only PSs, which is totally \emph{frequency-independent}. The corresponding transmit signal at subcarrier $m$ is given by $\tilde{\mathbf{x}}_m = \tilde{\mathbf{F}}_{\mathrm{RF}} \mathbf{F}_m^{\mathrm{BB}} \tilde{\mathbf{s}}_m$, where $\tilde{\mathbf{F}}_{\mathrm{RF}} \in \mathbb{C}^{N \times N_{\mathrm{RF}}}$ represents the \emph{frequency-independent} analog beamformer subject to unit-modulus constraints.
\end{itemize}   

\subsubsection{Spectral efficiency versus $P_t$}
In Fig. \ref{fig:wide_SE_power}, we plot the SE versus the maximum transmit power $P_t$ achieved by different schemes when $w=0$ (SE maximization) and $w=1$ (EE maximization). We set $M = 6 \times 6$. As we can see, there is a tradeoff between SE maximization and EE maximization in the high-power region, which is similar to the narrowband system. Furthermore, TTD-based hybrid beamforming achieves a comparable performance to fully-digital beamforming, while conventional hybrid beamforming causes significant performance loss in terms of SE. This is because by exploiting TTD-based hybrid beamforming, the impact of beam split at the BS and STARS can be efficiently reduced through the \emph{frequency-dependent} analog beamforming realized by TTDs. Finally, both independent and coupled phase-shift STARSs outperform conventional RIS in the wideband system.

\begin{figure}[t!]
    \centering
    \includegraphics[width=0.4\textwidth]{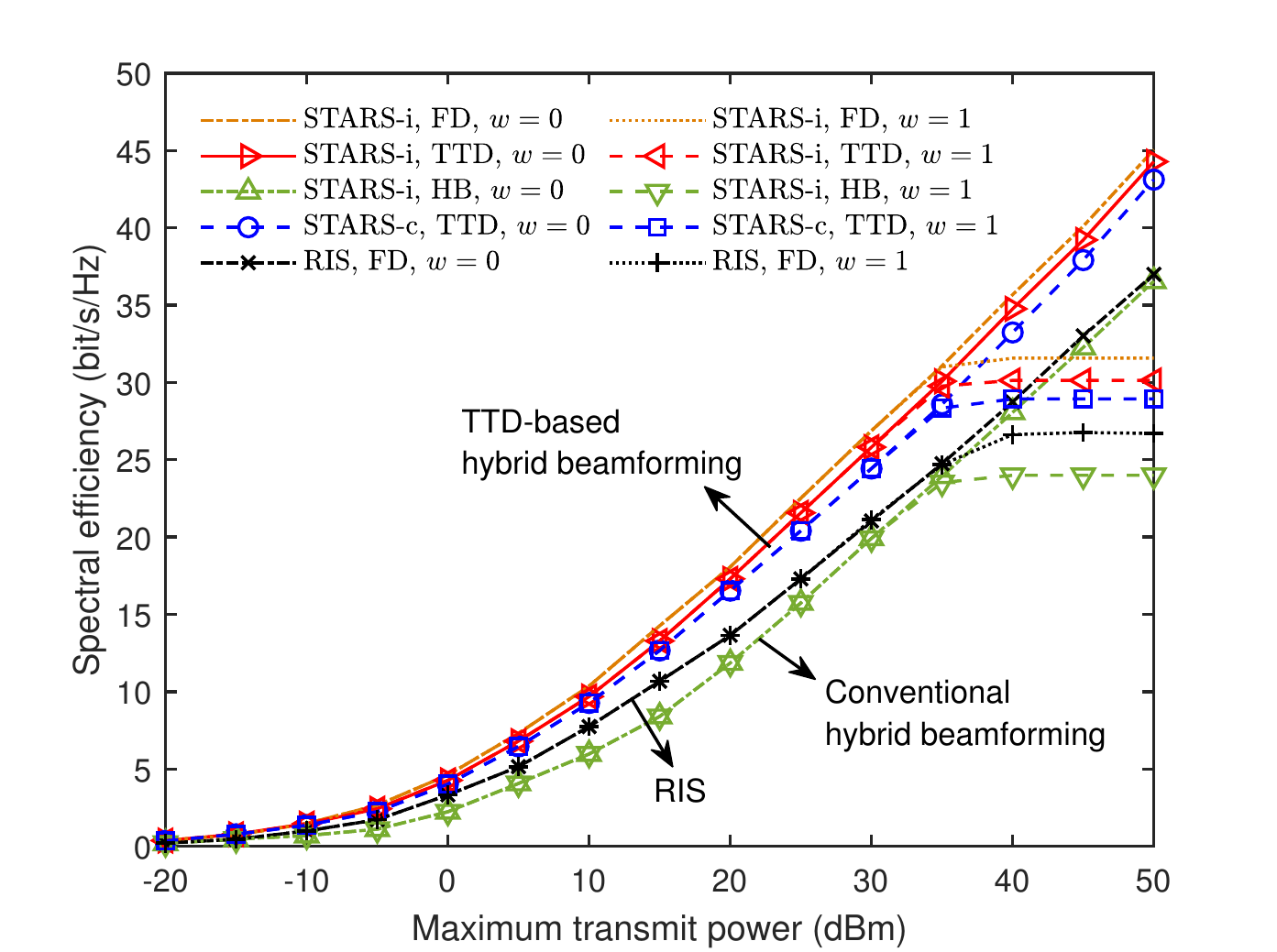}
    \caption{Spectral efficiency versus maximum transmit power $P_t$ for $M = 6 \times 6$ in the wideband system.}
    \label{fig:wide_SE_power}
\end{figure}
\begin{figure}[t!]
    \centering
    \includegraphics[width=0.4\textwidth]{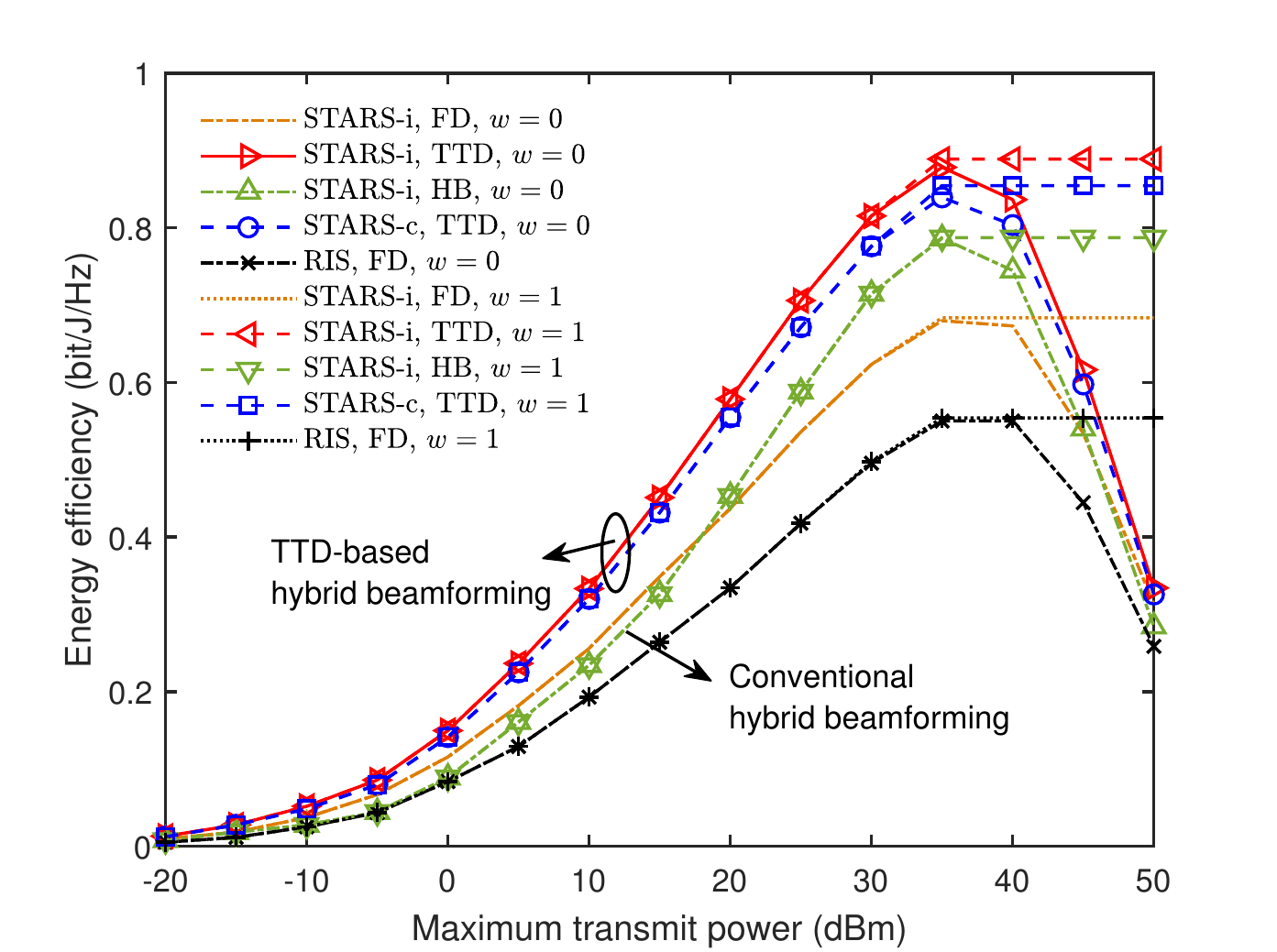}
    \caption{Energy efficiency versus maximum transmit power $P_t$ for $M = 6 \times 6$ in the wideband system.}
    \label{fig:wide_EE_power}
\end{figure}
\begin{figure}[t!]
    \centering
    \includegraphics[width=0.4\textwidth]{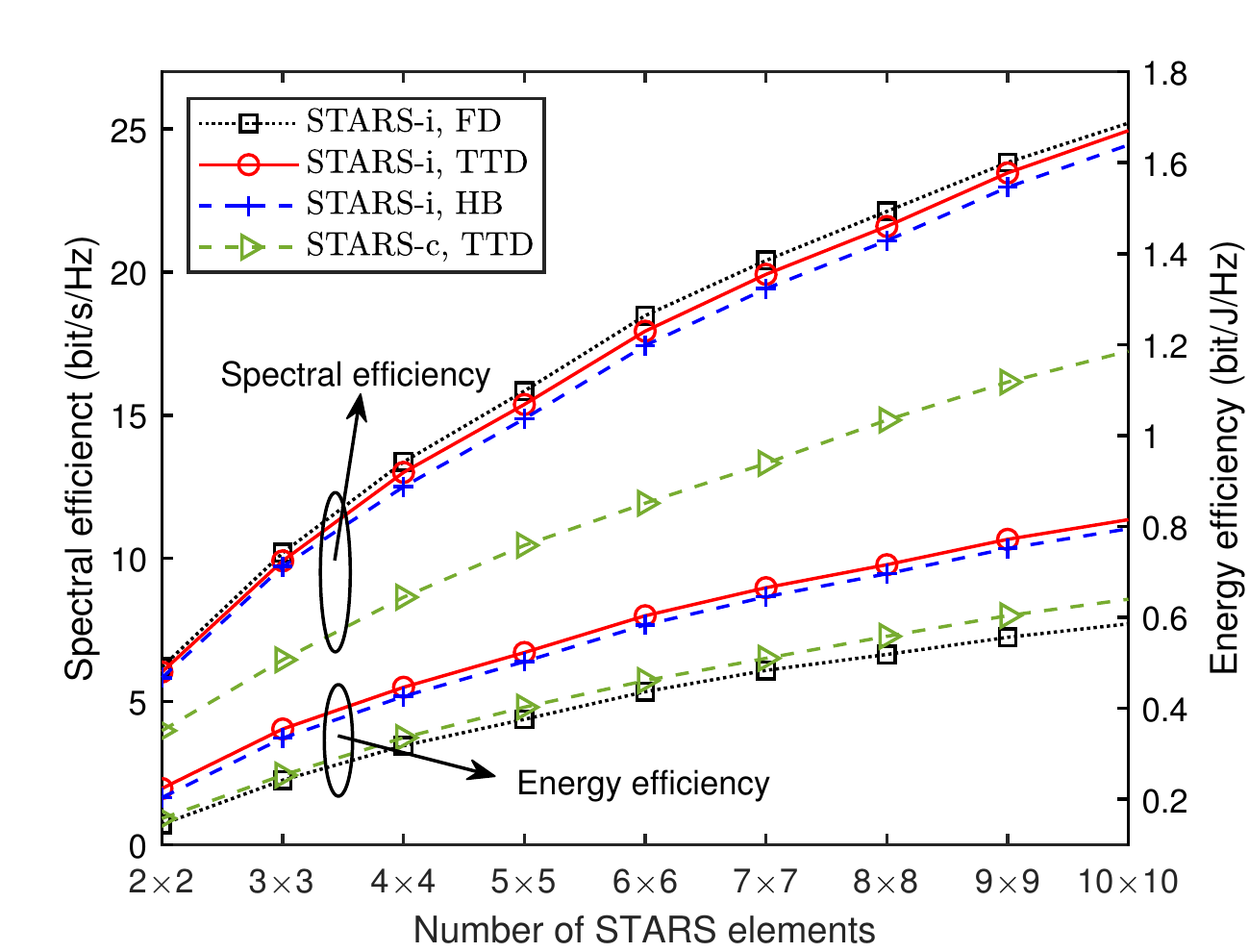}
    \caption{Performance versus the number of STARS elements for $P_t = 20$ dBm in the wideband system.}
    \label{fig:wide_EE_element}
\end{figure}

\subsubsection{Energy efficiency versus $P_t$}
In Fig. \ref{fig:wide_EE_power}, we further investigate the EE versus the maximum transmit power $P_t$ achieved by different schemes when $w=0$ (SE maximization) and $w=1$ (EE maximization). We set $M = 6 \times 6$. It can be observed that the proposed schemes that exploit STARS and TTD-based hybrid beamforming achieve the best performance in terms of EE. In particular, compared with conventional hybrid beamforming, although TTD-based hybrid beamforming involves additional TTDs that require much higher power consumption than PSs, it can significantly enhance SE and thus can realize higher EE.

\subsubsection{Impact of $M$}
In Fig. \ref{fig:wide_EE_element}, we illustrate the maximum SE and EE versus the number of STARS elements. We set $P_t = 20$ dBm. In wideband systems, the maximum SE and EE also increase with the number of STARS elements. Furthermore, it can also be observed that the performance gap between TTD-based hybrid beamforming and conventional hybrid beamforming becomes larger with the exploitation of more STARS elements. This is due to the larger array at the STARS which exacerbates the impact of beam split, further reducing the performance of conventional hybrid beamforming.
\section{Conclusion} \label{sec:conclusion}

A STARS-aided THz communication system was proposed. Considering both independent and coupled phase shifts, the power consumption models were proposed for STARSs. The general SE and EE optimization problems were formulated for both narrowband and wideband systems to jointly design the hybrid beamforming at the BS and the passive beamforming at the STARS. In particular, in wideband systems, the TTDs were introduced into the conventional hybrid beamforming structure to mitigate the wideband beam split. The numerical results confirmed the effectiveness of exploiting the STARS in THz communication systems. The great potential revealed in this work practical implementation of STARSs that can operate at THz frequencies, which could be an appealing future research direction.

\section*{Appendix~A\\Proof of Lemma \ref{lemma_2}} \label{appendiex:lemma_2}
First, we transform the objective function of problem \eqref{problem:narrow_subproblem_2} into the following more tractable form:
\begin{align}
   &\sum_{\chi \in \{t,r\}} \sum_{k \in \mathcal{K}_{\chi}} \| \mathbf{p}_k - \boldsymbol{\theta}_{\chi}^T \mathbf{H}_k^n \mathbf{F} + \rho \boldsymbol{\lambda}_k\|^2 \nonumber \\
    = & \sum_{\chi \in \{t,r\}} \underbrace{\boldsymbol{\theta}_{\chi}^H \mathbf{\Phi}_{\chi} \boldsymbol{\theta}_{\chi} - 2 \mathrm{Re}\{ \boldsymbol{\theta}_{\chi}^H \boldsymbol{\upsilon}_{\chi} \}}_{= g(\boldsymbol{\theta}_{\chi})} + C,
\end{align}
where 
\begin{align}
    \mathbf{\Phi}_{\chi} = &\sum_{k \in \mathcal{K}_{\chi}} (\mathbf{H}_k^n)^* \mathbf{F}^* \mathbf{F}^T (\mathbf{H}_k^n)^T, \\
    \boldsymbol{\upsilon}_{\chi} = &\sum_{k \in \mathcal{K}_{\chi}} (\mathbf{H}_k^n)^* \mathbf{F}^* (\mathbf{p}_k^T + \rho \boldsymbol{\lambda}_k^T),
\end{align}
and $C$ is the constant value irrelevant to $\boldsymbol{\theta}_{\chi}$. 
It is easily known that $g(\boldsymbol{\theta}_{\chi})$ is a quadratic function of each entry of $\boldsymbol{\theta}_{\chi}$. Let $\vartheta_{\chi, m}$ denote the $m$-th entry of $\boldsymbol{\theta}_{\chi}$. Then, the function $g(\boldsymbol{\theta}_{\chi})$ with respect to $\vartheta_{\chi, m}$ can be expressed in the following form:
\begin{equation} \label{new_obj_theta}
    \tilde{g}(\vartheta_{\chi, m}) = c_{\chi,m} |\vartheta_{\chi, m}|^2 -  2 \mathrm{Re} \{ d_{\chi,m}^* \vartheta_{\chi, m} \},
\end{equation}
where $c_{\chi,m}$ and $d_{\chi,m}$ are some real number and some complex number, respectively. The exact values of $c_{\chi,m}$ and $d_{\chi,m}$ can be obtained as follows. First, the derivative of $g(\boldsymbol{\theta}_{\chi})$ with respect to $\boldsymbol{\theta}_{\chi}$ is given by 
\begin{equation}
    \frac{\partial g(\boldsymbol{\theta}_{\chi})}{\partial \boldsymbol{\theta}_{\chi}^*} = \mathbf{\Phi}_{\chi} \boldsymbol{\theta}_{\chi} - \boldsymbol{\upsilon}_{\chi} .
\end{equation} 
The derivative of $\tilde{g}(\vartheta_{\chi, m})$ with respect to $\vartheta_{\chi, m}$ is given by 
\begin{equation}
    \frac{\partial \tilde{g}(\vartheta_{\chi, m})}{\partial \vartheta_{\chi, m}^*} =  c_{\chi,m} \vartheta_{\chi, m} - d_{\chi,m}.
\end{equation}
By comparing the above two derivatives, it can be obtained that $[\mathbf{\Phi}_{\chi} \boldsymbol{\theta}_{\chi} - \boldsymbol{\upsilon}_{\chi}]_{m} = c_{\chi,m} \vartheta_{\chi, m} - d_{\chi,m}$. Therefore, we have
\begin{subequations} \label{aux_theta}
    \begin{align}
        c_{\chi,m} = &[\mathbf{\Phi}_{\chi}]_{m,m}, \\
        d_{\chi,m} = &[\mathbf{\Phi}_{\chi}]_{m,m} [\boldsymbol{\theta}_{\chi}]_m  - [\mathbf{\Phi}_{\chi} \boldsymbol{\theta}_{\chi}]_m  + [\boldsymbol{\upsilon}_{\chi}]_{m}.
    \end{align}   
\end{subequations} 
Finally, by substituting $\vartheta_{\chi, m} = \beta_{\chi, m} e^{j \phi_{\chi, m} }$, where $\beta_{\chi,m} \in [0,1]$ and $\phi_{\chi,m} \in [0, 2\pi]$, into \eqref{new_obj_theta}, the objective function in \eqref{problem:narrow_subproblem_2} can be obtained. The proof is thus completed.

\bibliographystyle{IEEEtran}
\bibliography{reference/mybib}

\end{document}